%% file: main.tex
\documentclass[acmsmall,nonacm=true]{acmart}
\AtBeginDocument{%
  \providecommand\BibTeX{{%
    \normalfont B\kern-0.5em{\scshape i\kern-0.25em b}\kern-0.8em\TeX}}}

\usepackage{amsmath}
\usepackage{hyperref}
\usepackage{algpseudocode}
\usepackage{amsthm}
\usepackage{xcolor}
\usepackage{caption}
\usepackage{enumitem}
\usepackage{subcaption}
\usepackage{soul}
\usepackage{moresize}
\usepackage[linesnumbered,vlined,ruled,commentsnumbered]{algorithm2e}
\newcommand{\CT}{\textsf{CT}}
\newcommand{\JC}{\textsf{Join}}
\renewcommand{\P}{\mathbb{P}}
\newcommand{\Paragraph}[1]{\vspace{0.8 mm} \noindent {\bf #1}}

\definecolor{purple}{rgb}{0.74, 0.2, 0.64}

\newtheorem{theorem}{Theorem}[section]
\newtheorem{corollary}{Corollary}[theorem]
\newtheorem{lemma}[theorem]{Lemma}

\newcommand{\POB}{{\textsf{POB}}}

\newcommand{\ExactAlgName}[0]{\textcolor{black}{OCAP}}
\newcommand{\ApprAlgNameAbbr}[0]{\textcolor{black}{NOCAP}}

\begin{document}

\title{NOCAP: Near-Optimal Correlation-Aware Partitioning Joins}

\author{Zichen Zhu}
\email{zczhu@bu.edu}
\affiliation{%
\institution{Boston University}
 \city{Boston}
 \state{MA}
 \country{USA}
}

\author{Xiao Hu}
\email{xiaohu@uwaterloo.ca }
\affiliation{%
 \institution{University of Waterloo}
 \city{Waterloo}
 \state{Ontario}
 \country{Canada}
}

\author{Manos Athanassoulis}
\email{mathan@bu.edu}
\affiliation{%
 \institution{Boston University}
 \city{Boston}
 \state{MA}
 \country{USA}
}

\renewcommand{\shortauthors}{Zichen Zhu, Xiao Hu, and Manos Athanassoulis}

\ccsdesc[500]{Information systems~Join algorithms}
\ccsdesc[300]{Theory of computation~Database query processing and optimization (theory)}

\keywords{Storage-based Join, Partitioning Join, Dynamic Hybrid Hash Join}


\input{0-abstract}
\maketitle

\input{1-introduction}

\input{2-background}
\input{3-optimal-join}
\input{4-practical-join-algorithm}

\input{5-exp}
\input{6-discussion}

\input{7-related-work}
\input{8-conclusion}

\begin{acks}
This material is based upon work supported by the National Science Foundation under Grant No. IIS-2144547, a Facebook Faculty Research Award, and a Meta Gift.
\end{acks}

\newpage
\bibliographystyle{ACM-Reference-Format}
\bibliography{library,anon}

\appendix
\input{8-appendix}

\end{document}

%% file: 0-abstract.tex
\begin{abstract}
Storage-based joins are still commonly used today because the memory budget does not always scale with the data size. One of the many join algorithms developed that has been widely deployed and proven to be efficient is the Hybrid Hash Join (HHJ), which is designed to exploit any available memory to maximize the data that is joined directly in memory. However, HHJ cannot fully exploit detailed knowledge of the join attribute correlation distribution. 

In this paper, we show that given a correlation skew in the join attributes, HHJ partitions data in a suboptimal way. To do that, we derive the optimal partitioning using a new cost-based analysis of partitioning-based joins that is tailored for primary key - foreign key (PK-FK) joins, one of the most common join types.
This optimal partitioning strategy has a high memory cost, thus, we further derive an approximate algorithm that has tunable memory cost and leads to near-optimal results. Our algorithm, termed NOCAP (Near-Optimal Correlation-Aware Partitioning) join, outperforms the state-of-the-art for skewed correlations by up to $30\%$, and the textbook Grace Hash Join by up to $4\times$. Further, for a limited memory budget, NOCAP outperforms HHJ by up to 10\%, even for uniform correlation. Overall, NOCAP dominates state-of-the-art algorithms and mimics the best algorithm for a memory budget varying from below $\sqrt{\|\text{relation}\|}$ to more than $\|\text{relation}\|$.

\end{abstract}

%% file: 1-introduction.tex
\section{Introduction}
Joins are ubiquitous in database management systems (DBMS). Further, \emph{primary key - foreign key} (PK-FK) equi-joins are the most common type of joins. For example, all the queries of industry-grade benchmarks like TPC-H~\cite{TPCH} and most of Join Order Benchmark (JOB)~\cite{Leis2015} are PK-FK equi-joins. Recent research has focused on optimizing in-memory equi-joins~\cite{Boncz1999,Blanas2011,Wassenberg2011,Balkesen2013,Barber2014,Liu2015,Schuh2016,Bandle2021}, however, as the memory prices scale slower than storage \cite{McCallum2022}, the available memory might not always be sufficient to store both tables simultaneously, thus requiring a classical storage-based join \cite{Ramakrishnan2002}. 
This is common in a shared resource setting, like multiple colocated databases or virtual database instances deployed on the same physical cloud server~\cite{Jennings2015,Buyya2019,Berenberg2021}.
Besides, in edge computing,
memory is also limited, which is further exacerbated when other memory-demanding services are running~\cite{He2020}.
In several other data-intensive use cases like Internet-of-things, 5G communications, and autonomous vehicles~\cite{Cisco2018,Gartner2017}, memory might also be constrained.
Finally, there are two main reasons a workload might need to use storage-based joins: (1) workloads consisting of queries with multiple joins, and (2) workloads with a high number of concurrent queries. In both cases, the available memory and computing resources must be shared among all concurrently executed join operators.
\begin{figure}[t]
\centering
    \includegraphics[width=0.65\linewidth]{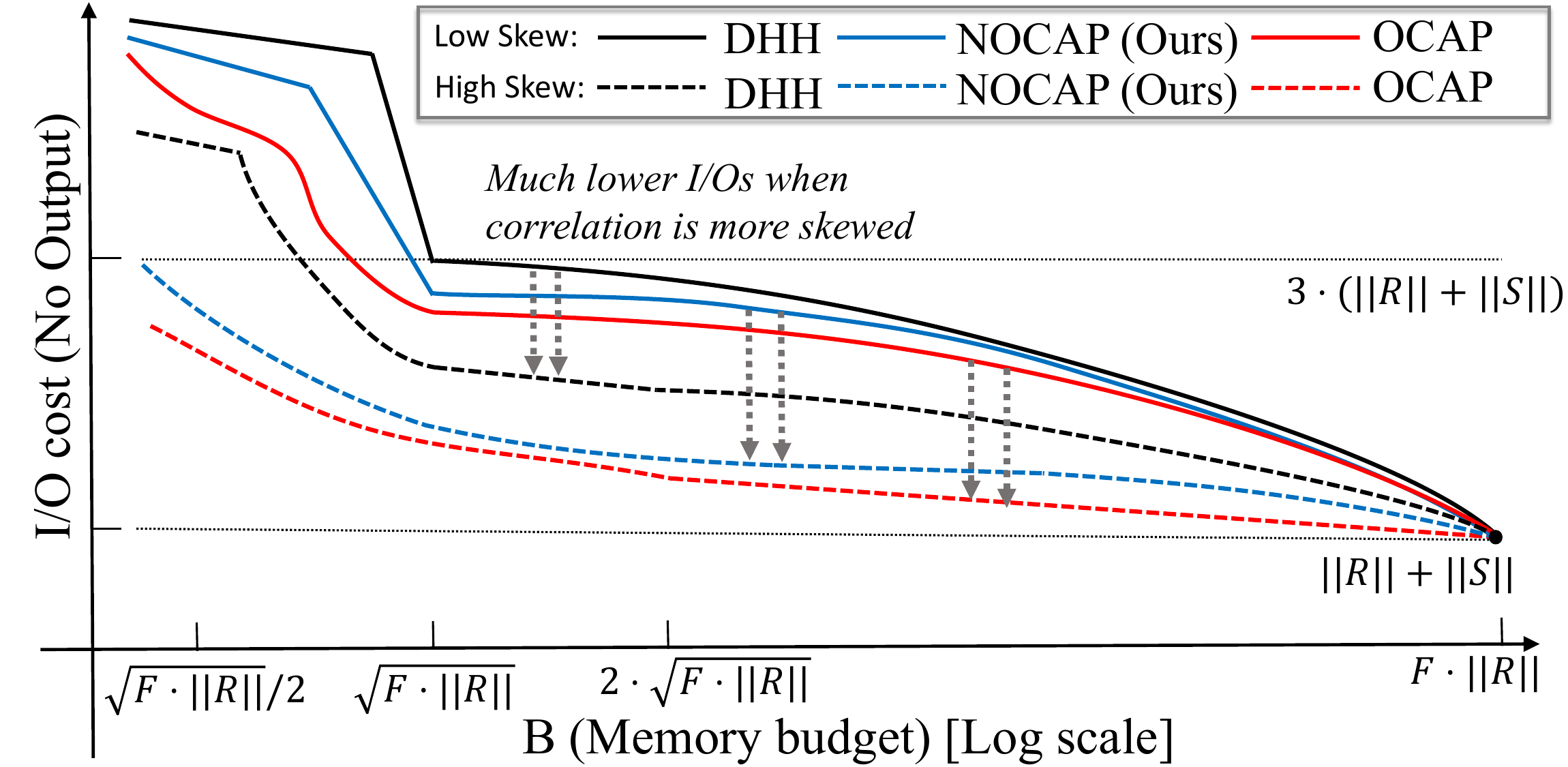}
    \label{fig:ghj-vs-smj-vs-ours-intro-uni}
        \vspace{-0.1in}
    \caption{A conceptual graph on the I/O cost of the state-of-the-art DHH, OCAP and NOCAP, assuming 
    $\|R\| \le \|S\|$. $F$ is the fudge factor indicating the space amplification factor between the in-memory hash table and the stored raw data).}
\label{fig:ghj-vs-smj-vs-ours-intro-1}
\end{figure}

\Paragraph{Storage-based Joins.} 
When executing a join, storage-based join algorithms are used when the available memory is not enough to hold the hash table for the smaller relation. 
Traditional storage-based join algorithms include Nested Block Join (NBJ), Sort Merge Join (SMJ), Grace Hash Join (GHJ), Simple Hash Join (SMJ), and Hybrid Hash Join (HHJ). 
Overall, HHJ, which acts as a blend of Simple Hash Join and Grace Hash Join, is considered the state-of-the-art approach~\cite{Kitsuregawa1983,DeWitt1984,Shapiro1986}, and is extensively used in existing database engines~(e.g., MySQL~\cite{MySQL}, PostgreSQL~\cite{PostgreSQL}, AsterixDB~\cite{Alsubaiee2014}).
HHJ \textbf{uniformly} distributes records from the two input relations to a number of partitions via hashing the join key and ensures that, when possible, one or more partitions remain in memory and are joined on the fly without writing them to disk. The remaining disk-resident partitions are joined using a classical hash join approach to produce the final result.
Unlike GHJ and SMJ, HHJ uses the available memory budget judiciously and thus achieves a lower I/O cost.
Existing relational database engines (e.g., PostgreSQL~\cite{PostgreSQL} and 
AsterixDB~\cite{Alsubaiee2014,Kim2020})
often implement a variant of HHJ, Dynamic Hybrid Hash join (DHH) that dynamically decides which partitions should stay in memory during partitioning.

\Paragraph{The Challenge: Exploit Skew.} For PK-FK joins, we describe the number of matching keys per PK using a \textit{distribution}, which we also refer to as \textit{join attribute correlation}, or simply \textit{join correlation}. In turn, join correlation can be uniform (all PKs have the same number of FK matches) or skewed (some PKs have more matches than others). Although there are many skew optimizations for joins~\cite{Hua1991,Hua1991a,Kitsuregawa1989,Nakayama1988,Li2002}, a potential detailed knowledge of the join attribute correlation and its skew is not fully exploited since existing techniques use heuristic rules to design the caching and partitioning strategy.
As such, while these heuristics may work well in some scenarios, in general, they lead to suboptimal I/O cost under an arbitrary join correlation. For example, HHJ can be optimized by prioritizing entries with high-frequency keys when building the in-memory hash table~\cite{PostgreSQLa}.
However, practical deployments use limited information about the join attribute correlation and typically employ a fixed threshold for building an in-memory hash table for keys with high skew (e.g., 2\% of available memory). As a result, prior work does not systematically quantify the benefit of such approaches, nor does it offer a detailed analysis of how close to optimal they might be.
In fact, due to the exponential search space, there is no previous literature that accurately reveals the relationship between the optimal I/O cost and the join correlation, thus leaving a large unexplored space for studying the caching and partitioning strategy.

\Paragraph{Key Idea: Optimal Correlation-Aware Partitioning.} To study the optimality of different partitioning-based join algorithms, we model a general partitioning strategy that allows for arbitrary partitioning schemes (not necessarily based on a specific hash function). 
We assume that the relations we join have a PK-FK relationship, and we develop an algorithm for \textit{Optimal Correlation-Aware Partitioning} (\ExactAlgName) that allows us to compare any partitioning strategy with the optimal partitioning, given the join correlation. Our analysis reveals that the state of the art is suboptimal in the entire spectrum of available memory budget (varied from $\sqrt{\|\text{relation}\|}$ to $\|\text{relation}\|$), leading up to $60\%$ more I/Os than strictly needed as shown in Figure \ref{fig:ghj-vs-smj-vs-ours-intro-1}, where the black lines are the state of the art, and the red lines are the optimal number of I/Os. The \ExactAlgName{} is constructed by modeling the PK-FK join cost as an optimization problem and then proving the \textit{consecutive theorem} which establishes the basis of finding the optimal cost within polynomial time complexity.
We propose a dynamic programming algorithm that finds the optimal solution in quadratic time complexity, $O(n^2\cdot m^2)$, and a set of pruning techniques that further reduce this cost to $O\left({\left(n^2\cdot \log m\right)}/m\right)$ (where $n$ is the number of records of the smaller relation and $m$ is the memory budget in pages). \ExactAlgName{} has a large memory footprint as it assumes that the detailed information of the join attributes correlation is readily available and, thus, can be only applied for offline analysis. 
Further, note that \ExactAlgName{} models the cost using the number of records as a proxy to I/Os, and is used as a building block of a practical algorithm that we discuss next.
We rely on \ExactAlgName{} to identify the headroom for improvement from the state of the art.

\Paragraph{The Solution: Near-Optimal Correlation-Aware Partitioning Join.} In order to build a practical join algorithm with a tunable memory budget, we approximate the optimal partitioning provided by \ExactAlgName{} with \textit{Near-Optimal Correlation-Aware Partitioning} (\ApprAlgNameAbbr{}) algorithm. \ApprAlgNameAbbr{} enforces a strict memory budget and splits the available memory between buffering partitions and caching information regarding skew keys in join correlation.
In Figure~\ref{fig:ghj-vs-smj-vs-ours-intro-1}, we construct a conceptual graph to compare our solution and the state-of-the-art storage-based join methods, such as Dynamic Hybrid Hash (DHH). As shown, DHH cannot fully adapt to the correlation and thus results in higher I/O cost than \ExactAlgName{} when the correlation is skewed. Our approximate algorithm, \ApprAlgNameAbbr{}, uses the same input from DHH and achieves near-optimal I/O cost compared with \ExactAlgName{}.
Further, while DHH is able to exploit the higher skew to reduce its I/O cost, the headroom for improvement for high skew is higher than for low skew, which is largely achieved by our approach. Overall, \ApprAlgNameAbbr{} is a practical join algorithm that is mostly beneficial compared to the state of the art, and offers its maximum benefit for a high skew in the join attribute correlation.

\Paragraph{Contributions.} In summary, our contributions are as follows:
\begin{itemize}[leftmargin=0.15in]  
\item We build a new cost model for partitioning-based PK-FK join algorithms, and propose optimal correlation-aware partitioning (\ExactAlgName{}) based on dynamic programming (\S\ref{subsec:opt_partition_no_caching}), assuming that the join correlation is known. \ExactAlgName{}'s time complexity is $O\left(n^2\cdot \log m / m\right)$ and space complexity is $O(n)$, where $n$ is the input size in tuples and $m$ is the memory budget in pages (\S\ref{subsec:hybrid-join}). 
\item We design an approximate correlation-aware partitioning (\ApprAlgNameAbbr{}) algorithm based on \ExactAlgName{}.
\ApprAlgNameAbbr{} uses partial correlation information (the same information used by the state-of-the-art skew optimized join algorithms) and achieves near-optimal performance while respecting memory budget constraints (\S\ref{sec:practical-partitioning-algorithm}).
\item We thoroughly examine the performance of our algorithm by comparing it against GHJ, SMJ, and DHH. We identify that the headroom for improvement is much higher for skewed join correlations by comparing the I/O cost of DHH against the optimal. Further, we show that \ApprAlgNameAbbr{} can reach near-optimal I/O cost and thus lower latency under different correlation skew and memory budget, compared to the state of the art (\S\ref{sec:exp}). Overall, \ApprAlgNameAbbr{} dominates the state-of-the-art and offers performance benefits of up to 30\% when compared against DHH and up to $4\times$ when compared against the textbook GHJ.
\end{itemize}

%% file: 2-background.tex
\section{Previous Storage-based Joins}
We first review four classical storage-based join methods~\cite{DeWitt1984,Haas1997,Jahangiri2022,Ramakrishnan2002} (see Table~\ref{tab:cost_model}), and then present more details of Dynamic Hybrid Hash (DHH), a variant of the state-of-the-art HHJ. 
\subsection{Classical Storage-based Joins}
\label{subsec:classical_joins}
\Paragraph{Nested Block Join (NBJ).} NBJ partially loads the smaller relation $R$ (in \textbf{chunks} equal to the available memory) in the form of an in-memory hash table and then scans the larger relation $S$ once \emph{per chunk} to produce the join output for the partial data. 
This process is repeated multiple times for the smaller relation until the entire relation is scanned.
As such, the larger relation is scanned for as many times as the number of chunks in the smaller relation.

\Paragraph{Sort Merge Join (SMJ).} SMJ works by sorting both input tables by the join attribute using external sorting and applying $M$-way ($M\leq B-1$) merge sort to produce the join result.
During the external sorting process, if the number of total runs is less than $B-1$, the last sort phase can be combined together with the multi-way join phase to avoid repetitive reads and writes \cite{Ramakrishnan2002}.

\Paragraph{Grace Hash Join (GHJ).} GHJ uniformly distributes records from the two input relations to a number of partitions (at most $B-1$) via hashing the join key, and the corresponding partitions are joined then. Specifically, when the smaller partition fits in memory, we simply store it in memory (typically as a hash table) and then scan the larger partition to produce the output .

\begin{table}[h]
\caption{Estimated Cost for NBJ, GHJ, and SMJ for $R \Join S$. $\|R\|$ and $\|S\|$ are the number of pages of $R$ and $S$. Assume $\|R\| \le \|S\|$. \#chunks is the number of passes for scanning $S$. 
\#pa-runs is the number of times to partition $R$ and $S$ until the smaller partition fits in memory.
\#s-passes is the number of partially sorted passes of $R$ and $S$, until the number of the total sorted runs is $\le B-1$ ($B$ is the total number of memory available). $\mu$ and $\tau$ indicate the write/read asymmetry, where $RW$, $SW$, and $SR$ denote the latency per I/O for random write, sequential write and sequential read respectively.}
\vspace{-0.1in}
\begin{tabular}{ccc}
	\toprule
	Approach & Normalized \#I/O & Notation \\
	\midrule
	$\textsf{NBJ}(R,S)$ & $ \|R\|+ \#\text{chunks}\cdot \|S\| $ & None\\
        $\textsf{GHJ}(R,S)$ & $\left(1 + \#\text{pa-runs}\cdot \left(1+\mu\right)\right)\cdot \left(\|R\|+\|S\|\right)$ & $\mu \stackrel{\text{def}}{=} RW/SR$ \\
	  $\textsf{SMJ}(R,S)$ & $\left(1 + \#\text{s-passes}\cdot \left(1+\tau\right)\right)\cdot \left(\|R\|+\|S\|\right)$ & $\tau \stackrel{\text{def}}{=} SW/SR$\\
	\bottomrule
\end{tabular}
	\label{tab:cost_model}	
\end{table}

\Paragraph{Hybrid Hash Join (HHJ).} HHJ is a variant of GHJ that allows one or more partitions to stay in memory without being spilled to disk, if the space is sufficient. When partitioning the second relation, we can directly probe in-memory partitions and generate the join output for those in-memory partitions.
The remaining keys of the second relation are partitioned to disk, and we execute the same probing phase as in GHJ to join the on-disk partitions.
When the memory budget is lower than $\sqrt{\|R\|\cdot F}$ (where $\|R\|$ is the size of the smaller relation in pages, and $F$ is the fudge factor that stands for the space amplification factor between the in-memory hash table and the stored raw data), HHJ downgrades into GHJ because it will not be feasible to maintain a partition in memory while having enough buffers for the remaining partitions and the output.

\begin{figure}[t]
\includegraphics[width=0.75\textwidth]{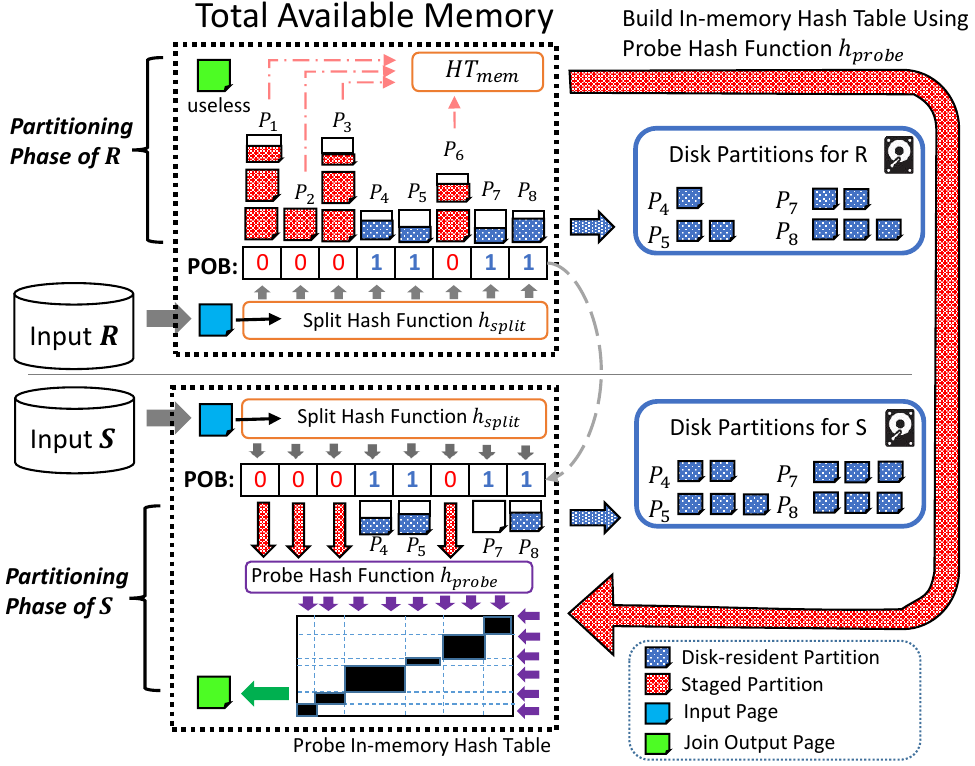}
\caption{An illustration of partitioning in DHH.
Partitions in red are staged in memory when $R$ is being partitioned, while the remaining ones in blue are written to disk.}
\label{fig:dynamic_hybrid_join_intro}
\end{figure}

\begin{algorithm}[t]
\caption{$\textsc{DHH-Primary}(R, HS_{\textsf{mem}}, B)$}\label{alg:DHH_alg_R}
$m_{\textsf{DHH}} \gets \max\left\{20, \left\lceil \frac{\|R\|\cdot F - B}{B-1}\right\rceil\right\}$\;
Initialize a length-$m_{\textsf{DHH}}$ array $\POB$ as $\{0,0,...,0\}$\;
Initialize a length-$m_{\textsf{DHH}}$ partitioning $\textsf{PartPages}$\;
\ForEach{$r \in R$}{ 
    $\textsf{PartID} \gets h_{\textsf{split}}(r.\textsf{key})$\;
    \If{{\normalfont $\POB[{\normalfont \textsf{PartID}}] = 0$}}{
        \If{{\normalfont $\sum\limits_{i=1}^{m_{\textsf{DHH}}} \textsf{PartPages}[i].\textsf{size}() = B - 2 - m_{\textsf{DHH}}$}}{
            Flush an arbitrary partition $j$ to disk\;
            {\normalfont $\POB[j]=1$}\;
        } 
        \lElse{append $r$ to $\textsf{PartPages}[\textsf{PartID}]$} 
    }
    \lElse{append $r$ to the in-disk partition $\textsf{PartID}$}     
}
\lForEach{{\normalfont $r \in \textsf{PartPages}$}}{add $r$ to $HT_{\textsf{mem}}[h_\textsf{probe}(r.\textsf{key})]$}
\Return $HT_{\textsf{mem}}, \POB$\;
\end{algorithm}
\begin{algorithm}[t]
\caption{$\textsc{DHH-Foreign}(S, HT_{\textsf{mem}}, \POB$)}\label{alg:DHH_alg_S}
\ForEach{$s \in S$}{ 
    \If{{\normalfont $h_\textsf{probe}(s.\textsf{key}) \in HT_{\textsf{mem}}$}}{
        \If{\normalfont $s.\textsf{key}$ is in records $HT_{\textsf{mem}}[h_\textsf{probe}(s.\textsf{key})]$}{
          join $s$ with matched records and output\;
        }
        
    } \Else{
        \If{{\normalfont $\POB[h_{\textsf{split}}(s.\textsf{key})] = 1$}} {
            append $s$ to the in-disk partition 
        $\textsf{PartID}$\;
        }
    }   
}
\end{algorithm}

\subsection{Dynamic Hybrid Hash Join}
\label{subsec:dynamic_hhj}
Here we review DHH, the state-of-the-art implementation of HHJ, as well as its skew optimization. While there are many different implementations of DHH~\cite{Jahangiri2022}, the core idea is the same across all of them: decide which partitions to stay in memory or on disk.

\noindent {\bf Framework.}
We visualize one implementation  (AsterixDB~\cite{Alsubaiee2014,Kim2020}) of DHH in Figure~\ref{fig:dynamic_hybrid_join_intro}.
Compared to HHJ, DHH exhibits higher robustness by dynamically deciding which partitions should be spilled to disk.
Specifically, every partition is initially staged in memory, and when needed, a staged partition will be selected to be written to disk, and it will free up its pages for new incoming records.
Typically, the largest partition will be selected, however, the victim selection policy may vary in different systems.
A bit vector (``Page Out Bits'', short as POB, in Figure~\ref{fig:dynamic_hybrid_join_intro}) is maintained to record which partitions have been written out at the end of the partitioning phase of $R$. Another hash function is applied to build a large in-memory hash table using all staged partitions. 
It is worth noting that, when no partitions can be staged in the memory (i.e., all the partitions are spilled out to disk), DHH naturally \textbf{downgrades} into GHJ.
When partitioning the relation $S$, DHH applies the partitioning hash function first and checks POB for every record. 
If it is set, DHH directly probes the in-memory hash table and generate the join result. 
Otherwise, DHH moves the joined record to the corresponding output page and flushes it to the disk if it is full. 
Finally, DHH performs exactly the same steps as the traditional hash join to probe all the disk-resident partitions. The partitioning and probing of DHH is presented in detail in Algorithms \ref{alg:DHH_alg_R} and \ref{alg:DHH_alg_S}.

\Paragraph{Number of Partitions.} The number of partitions in DHH (noted by $m_{\textsf{DHH}}$) is a key parameter that largely determines join performance.
A large $m_{\textsf{DHH}}$ requires more output buffer pages reserved for each partition, and, thus, less memory is left for memory-resident partitions.
On the other hand, small $m_{\textsf{DHH}}$ renders the size of each partition very large, which makes it harder to stage in memory.
No previous work has formally investigated the optimal number of partitions for arbitrary correlation skew, but past work on HHJ employs some heuristic rules.
For example, an experimental study~\cite{Jahangiri2022} recommends that 20 is the minimum number of partitions when we do not have sufficient information about the join correlation.
In addition, if we restrict the number of memory-resident partitions to be 1, we can set the number of partitions $m_{\textsf{DHH}}$ as
$m_{\textsf{DHH}} = \left\lceil \frac{\|R\|\cdot F - B}{B-1}\right\rceil$~\cite{Shapiro1986},
where $\|R\|$ is the size of relation $R$ in pages, $F$ is the fudge factor for the hash table, and $B$ is the total given memory budget in pages.
Integrating this equation with the above tuning guideline, we configure $m_{\textsf{DHH}}=\max\left\{\left\lceil \frac{\|R\|\cdot F - B}{B-1}\right\rceil, 20\right\}$.

\Paragraph{Heuristic Skew Optimization In Practical Deployments.}
For ease of notation, we assume that $R$ is the relation that holds the primary key and is smaller in size (dimension table) and that $S$ is the relation that holds the foreign key and is larger in size (fact table).
In PK-FK joins, the join correlation can be non-uniform. Histojoin~\cite{Cutt2009} proposes to cache in memory the Most Common Values (MCVs, maintained by DBMS) to help reducing the I/O cost.
Specifically, if the memory budget is enough, Histojoin first partitions the dimension table $R$ and caches entries with high-frequency keys in a small hash table in memory.
When it partitions the fact table $S$, if there are any keys matched in the small hash table, it directly joins them and output.
That way, Histojoin avoids writing out many entries with the same key from the fact table and reading them back into memory during the probing phase.
However, it is unclear how much frequency should be treated as \textit{high} to qualify for this optimization and how large the hash table for the skewed keys should be within the available memory budget. 
In one implementation of Histojoin~\cite{Cutt2008}, it limits the memory usage to 2\% for constructing the hash table specifically for skewed keys. This restriction ensures that the specific hash table fits in memory.
Likewise, in PostgreSQL~\cite{PostgreSQLa}, when the skew optimization knob is enabled and MCVs are present in the system cache (thus not occupying the user-specified memory), PostgreSQL follows a similar approach by adding an extra trigger to build the hash table for skewed keys.
By default, PostgreSQL also allocates a fixed budget of 2\% memory for the hash table dedicated to handling skewed keys, but this allocation only occurs when the total frequency of skewed keys is more than $0.01\cdot n_S$ ($n_S$ is the number of keys in relation $S$).
In fact, PostgreSQL implements a general version of Histojoin by specifying the frequency threshold.
However, existing skew optimizations in both PostgreSQL and Histojoin rely on fixed thresholds (i.e., memory budget for skew keys, and frequency threshold), which apparently cannot be adapted to different memory budget and different join skew. For example, as shown in Figure \ref{fig:ghj-vs-smj-vs-ours-intro-1}, state-of-the-art DHH are not enough to achieve ideal I/O cost, and they are further away from the optimal for higher skew.

%% file: 3-optimal-join.tex
\section{Optimal Correlation-Aware Partitioning Join}
\label{sec:correlation-aware_join}
\begin{table}[t]
\caption{Summary of our notation.}
\vspace{-0.1in}
\begin{tabular}{cc}
	\toprule
	Notation & Definitions (Explanations) \\
	\midrule
	 $\|R\|$ $(\|S\|)$ & \#pages of relation $R$ ($S$) \\
  $n_R$ $(n_S)$ & \#records from relation $R$ ($S$) \\
  $P_j$ & the $j$-th partition ($P_j = R_j$)\\
  $|P_j|$ & \#records in $j$-th partition \\
  $\|R_j\|$ $(\|S_j\|)$ & \#pages of a partition $R_j$ ($S_j$) \\
     $\P$ & a Boolean matrix for partitioning assignation \\
     $f: i \rightarrow j$ & returns $j$ for $i$ so that $\P_{i,j}=1$  \\
     $\mathcal{N}_f(i)$ & the partition where the $i$-th records belongs\\
     $n$ & $n = n_R$ in the partitioning context \\
     $m$ & \#partitions \\
     
     $B$ & \#pages of the total available buffer \\
     $F$ & the fudge factor of the hash table \\
     $\CT$ & the correlation table of keys and their frequency \\
    $b_R$ ($b_S$) & \#records from relation $R$ ($S$) per page \\
    $c_R$ ($c_S$) & \#records from relation $R$ ($S$) per chunk in $\textsf{NBJ}$ \\   
	\bottomrule
\end{tabular}
\label{tab:notation}
\end{table} 

We now discuss the theoretical limit of an I/O-optimal partitioning-based join algorithm when the correlation skewness information is known in advance and can be accessed for free. We assume that $R$ is the relation that holds the primary key (dimension table) and $S$ is the relation that holds the foreign key (fact table). More specifically, we model the correlation between two input relations with respect to the join attributes as a \emph{correlation table} $\CT$, where $\CT[i]$ is the number of records in $S$ matching the $i$-th record in $R$. Our goal is to find which keys should be cached in memory, and how the rest of the keys should be partitioned in order to minimize the total I/Os of the join execution with an arbitrary memory budget.
 
We start with an easy case when no records can be cached during the partition phase (\S\ref{subsec:opt_partition_no_caching}) 
and turn to the general case when keys can be cached (\S\ref{subsec:hybrid-join}).
We note again that the memory used for storing $\CT$ as well as the optimal partitioning does not compete with the available memory budget $B$, which is indeed \textbf{unrealistic}. This is why we consider it as a theoretically I/O-optimal algorithm, since it can help us understand the lower bound of any feasible partitioning but cannot be converted into a practical algorithm. More proofs and details can be found in the appendix.

\subsection{Optimal Partitioning Without Caching}
\label{subsec:opt_partition_no_caching}
We first build the cost model for partitioning assuming no cached records during the partitioning (\S\ref{subsec:modeling_io}), then present the main theorem for characterizing the optimal partitioning (\S\ref{subsec:consecutive_theorem}) and find the optimal partitioning efficiently via dynamic programming (\S\ref{subsec:dynamic_programming}).

\subsubsection{Partitioning as An Integer Program}
\label{subsec:modeling_io}

A partitioning $\P$ is an assignment of $n$ records from $R$ to $m$ partitions. As we use at least one page to stream the input relation and the rest of the pages as output buffers for each partition, $m \leq B-1$. More specific constraints on $m$ will be determined in Section~\ref{subsec:hybrid-join} and Section~\ref{subsec:approximate_algorithm}. 
Once we know how to partition $R$ (the relation with primary key), we can apply the same partition strategy to $S$.

\Paragraph{Partitioning.} We encode a partitioning $\P$ as a Boolean matrix of size $n \times m$, where $\P_{i,j}=1$ indicates that the $i$-th record belongs to the $j$-th partition. 
If the $i$-th record from $R$ does not have a match in $S$, i.e., $\CT[i]=0$, we will not assign it to any partition. We can preprocess the input to filter out these records so that each of remaining records belongs to exactly one partition, although it is unrealistic in practice.
In this way, a partitioning $\P$ can be compressed as a mapping function $f: [n] \to [m]$ from the index of record to the index of partition in $\P$, such that $i$-th record from relation $R$ is assigned to $f(i)$-th partition in $\P$. 
 Let $P_j = \{i \in [n]: f(i)=j \}$ be the set of indexes of records from $R$ assigned to the $j$-th partition in $\P$. We can then derive the number of pages for $R_j$ and $S_j$ as $\|R_j\|= \left\lceil |P_j|/b_R \right\rceil$, $\|S_j\|= \lceil \sum_{i \in P_j}\CT[i]/b_S \rceil$, where $b_R$ (resp. $b_S$) is the number of records per page for $R$ (resp. $S$).  Moreover, we use $\mathcal{N}_f(i) =\{i' \in [n]: f(i') = f(i)\}$ to denote the records from $S$ that fall into the same partition with the $i$-th record in $R$.
We summarize frequently-used notations in Table~\ref{tab:notation}.

\Paragraph{Cost Function.} Given a partitioning $\mathbb{P}$, it remains to join partitions in a pair-wise manner.
Reusing the cost model in Table~\ref{tab:cost_model}, the cost of computing the join results for a pair of partitions $(R_j, S_j)$ is $\min \{\textsf{NBJ}(R_j,S_j), \textsf{SMJ}(R_j,S_j), \textsf{GHJ}(R_j,S_j)\}$, by picking the cheapest one among \textsf{NBJ}, \textsf{SMJ}, and \textsf{GHJ}. 
We omit \textsf{DHH} here for simplicity, which does not lead to major changes of the optimal partitioning, since in most cases \textsf{NBJ} is the most efficient one when taking read/write asymmetry into consideration (write is generally much slower than read in modern SSDs~\cite{Papon2021,Papon2021a}).
For ease of illustration, we assume only \textsf{NBJ} is always applied in the pair-wise join, i.e., the \emph{join cost} induced by partitioning $\P$ is $\JC(\P,m)= \sum_{j=1}^{m}\textsf{NBJ}(R_j, S_j)$. It can be proved that most of the following results still apply when alll of \textsf{NBJ}, \textsf{SMJ}, and \textsf{GHJ} can be used for the partition-wise joins as long as we always pick the one with minimum cost.

Recall the cost function of $\textsf{NBJ}$ in Table~\ref{tab:cost_model}. For simplicity, we assume $\|R_j\| \le \|S_j\|$ for each $j \in [m]$. The other case with $\|R_j\| \ge \|S_j\|$ can be discussed similarly. By
setting \#chunks$= \lceil \|R_j\|/((B-2)/F) \rceil$, $\|R_j\| = |P_j|/b_R$ and $\|S_j\|= \lceil \sum_{i \in P_j}\CT[i]/b_S \rceil$, we obtain the cost function as follows:
\begin{equation*}
\begin{split}
\JC(\P,m) 
=\|R\| + \sum\limits_{j=1}^{m} \left\lceil \frac{|P_j|}{c_R} \right\rceil  \cdot 
\left\lceil \sum\limits_{i \in P_j}\frac{\CT[i]}{b_S}\right\rceil
&\le \|R\| + m + \frac{n_R}{c_R} + \sum\limits_{j=1}^{m} \left\lceil \frac{|P_j|}{c_R} \right\rceil  \cdot 
\sum\limits_{i \in P_j}\frac{\CT[i]}{b_S}\\
& = \|R\| + m + \frac{n_R}{c_R} + \frac{1}{b_S} \cdot \sum\limits_{i=1}^{n}\CT[i] \cdot \left\lceil \frac{|\mathcal{N}_f(i)|}{c_R} \right\rceil
\end{split}
\label{eq:join_cost}
\end{equation*}
where $c_R=\lfloor b_R\cdot (B-2)/F\rfloor$ is a constant that denotes the number of records per chunk in relation $R$, and $\mathbb{I}$ is the indicator function. 
Note that the I/O cost of read and write both relations is the same for all partitioning strategies. 

\Paragraph{Integer Program.} Given the number of partitions $m$, finding a partitioning $\P$ for $R$ (containing the primary key) of $n$ records with minimum cost can be captured by the following integer program: 
\begin{align*}
\min_{\P} & \quad  \JC(\P,m)\\
\textrm{subject to} & \quad \sum\limits_{j=1}^{m}\P_{i,j} = 1, \forall i \in [n] \\
& \quad m + 1 \le B \\
& \quad \P_{i,j} \in \{0,1\}, \forall i \in [n], \forall j \in [m]
\end{align*}
However, solving such an integer program takes exponentially large time by exhaustively searching all possible partitionings. We next characterize the optimal partitioning with some nice properties. 

 \subsubsection{Main Theorem}
 \label{subsec:consecutive_theorem}
 
For simplicity, we assume that $\CT$ is \emph{sorted in ascending order}, i.e., if $i_1 \le i_2$, then $\CT[i_1] \le \CT[i_2]$. With sorted $\CT$, we present our main theorem for characterizing the optimal partitioning $\P$ with three critical properties. Firstly, each partition of $\P$ contains consecutive records lying on the sorted \CT~ array. 
Then, we can define a partial ordering on the partition in $\P$. Secondly, the sizes of all partitions are in a weakly descending order. At last, the size of each partition except the first one (or the largest one) is divisible by $c_R$. 
\begin{theorem}[Main Theorem]
\label{the:main}
Given an arbitrary sorted {\normalfont \CT} array, 
there is an optimal partitioning $\P = \langle P_1,P_2,\cdots,P_m \rangle$ such that 
\begin{itemize}[leftmargin=*]
    \item {\bf (consecutive property)} for any $i_1 \leq i_2$, if $i_1 \in P_j$ and $i_2 \in P_j$, then $i \in P_j$ holds for any $i \in [i_1, i_2]$;
    \item {\bf (weakly-ordered property) $\left\lceil \frac{|P_1|}{c_R} \right\rceil \ge \left\lceil \frac{|P_2|} {c_R} \right\rceil  \ge \cdots \ge \left\lceil \frac{ |P_m|}{c_R} \right\rceil$};
    
    \item {\bf (divisible property)} for any $2\le i \le m$, $|P_i|$ is divisible by $c_R$.
\end{itemize}
where $P_j$ is the set of records assigned to the $j$-th partition of $\P$, i.e., $P_j = \{i \in [n]: f(i) =j\}$.
\end{theorem}

\begin{lemma}[Swap Lemma]
\label{lem:swap}
For a partitioning $\P$, if there exists a pair of indexes $i_1,i_2$ with {\normalfont $\CT[i_1] > \CT[i_2]$} and $|\mathcal{N}_f(i_1)| \geq |\mathcal{N}_f(i_2)|$, it is feasible to find a new partitioning $\P'$  by swapping records at $i_1$ and $i_2$, such that {\normalfont $\JC(\P) \ge \JC(\P')$}.
\end{lemma}

\begin{figure}[t]
\includegraphics[width=0.65\textwidth]{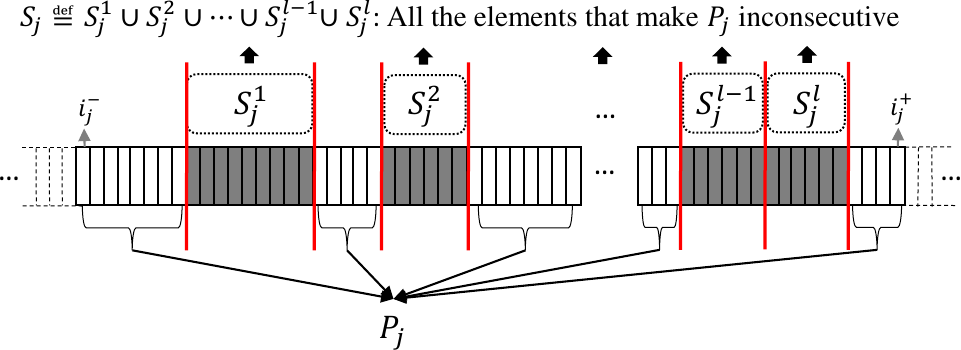}
\caption{An illustration of records from one partition $P_j$ lying in a non-consecutive way in the sorted $\CT$.}
 \label{fig:inconsecutive-P1}
 \end{figure}
To prove the main theorem, the high-level idea is to transform an arbitrary partitioning $\P$ into $\P'$ with desired properties, while keeping the join cost of $\P'$ always smaller (or at least no larger) than that of $\P$. The essence is a {\em swap} procedure, as described by Lemma~\ref{lem:swap}, which can reassign records in a way such records with larger \CT~ values always appear with fewer records in a partition. Its proof directly follows the definition of join cost. Next, we will show how to transform an arbitrary partitioning $\P$ step by step.

\Paragraph{Consecutive Property.} Given an arbitrary partitioning $\P$, we first show how to transform it into $\P'$ such that $\P'$ satisfies the consecutive property while $\JC(\P') \le \JC(\P)$. As described in Algorithm~\ref{alg:swap}, we always identify a partition whose records are not consecutive on the sorted $\CT$, say $P_j$, and apply the swap procedure to it.
Let $i^-_j$ and $i^+_j$ be the minimum and maximum index of records from $P_j$ lying on $\CT$. Let $S^1_j, S^2_j, \cdots, S^\ell_j$ be the longest consecutive records whose indexes fall into $[i^-_j, i^+_j]$ in $\CT$, and belong to the same partition $P_{j'}$ for some $j' \neq j$ ($S^{\ell_1}_j$ and $S^{\ell_2}_j$ may belong to  different partitions). 
 We denote $S_j = S^1_j \cup  S^2_j \cup \cdots \cup S^\ell_j$ as the set of all records that make $P_j$ non-consecutive with respect to the sorted $\CT$ (Figure~\ref{fig:inconsecutive-P1}).
 \begin{algorithm}[h]
\caption{\textsc{SWAP}$(\P)$}
 \label{alg:swap}
 \While{there exists $P_j \in \P$ that is non-consecutive}{
 \While{{\normalfont \bf{true}}}{
    $\textsf{left}, \textsf{right} \gets \min_{i \in [n]: f(i) = j} i, \max_{i \in [n]: f(i) = j}i$\; 
    $S_j \gets \{i \in [n]: \textsf{left} \le i \le \textsf{right}, f(i) \neq j\}$\;
    \lIf{$S_j = \emptyset$}{\textbf{break}}
    $\ell \gets \min_{i \in S_j} i$\;
    \If{$|\mathcal{N}_f(\ell)| > |P_j|$}{
        $\textsc{swap}(f(\ell), f(\textsf{left}))$, $\textsf{left} \gets \textsf{left} + 1$\;
    }
    \lElse{$\textsc{swap}(f(\ell), f(\textsf{right}))$, $\textsf{right} \gets \textsf{right} -1$}
    }
 }
 \Return $\P$\;
\end{algorithm}

Algorithm~\ref{alg:swap} maintains two pointers $\textsf{left}$ and $\textsf{right}$ to indicate the smallest and largest indexes of records to be swapped. Consider the smallest index $\ell \in S_j$. 
 If $|\mathcal{N}_f(i)| > |P_j|$, we swap $\ell$ and $\textsf{left}$ by putting the record at $i$ into $P_{f(\textsf{left})}$ and the record at $\textsf{left}$ into $P_{f(\ell)}$ ($P_{f(\ell)} = \mathcal{N}_f(\ell)$ by definition). Otherwise, we swap $\ell$ and $\textsf{right}$ similarly. 
 We repeat the procedure above until all partitions become consecutive on the sorted $\CT$. It can be seen that the resulted partitioning from $\textsc{SWAP}(\P)$ is ``consecutive'', and the cost does not increase in this process, thanks to the swap lemma.

\begin{algorithm}[h]
 \caption{\textsc{ORDER}$(\P)$}
 \label{alg:sort}
 \While{there exists $P_{j'} \precsim P_j \in \P$ s.t. $\left\lceil \frac{|P_{j'}|}{c_R} \right\rceil < \left\lceil \frac{|P_j|}{c_R} \right\rceil$ }{
    $\ell, \ell' \gets \min_{i \in P_j} i,  \min_{i \in P_{j'}} i$\;
    \lFor{$i \gets 0$ to $|P_{j'}|$}{\textsc{swap}$(f(\ell + i), f(\ell' + i))$}
 $\P \gets \textsc{SWAP}(\P)$\;
 }
\end{algorithm}
 
\Paragraph{Weakly-Ordered Property.} Given a partitioning $\P$ with consecutive property, we next show how to transform it into $\P'$ such that $\P'$ satisfies both consecutive and weakly-ordered properties.
As each partition in $\P$ contains consecutive records in the sorted $\CT$, we can define a partial ordering $\precsim$ on the partitions in $\P$.
For simplicity, we assume that $P_1 \precsim P_2 \precsim \cdots \precsim P_m$.
As described in Algorithm~\ref{alg:sort}, we always identify a pair of partitions that are not ``weakly-ordered'', i.e., $P_{j'} \precsim P_j$ but $\left\lceil |P_{j'}|/c_R \right\rceil < \left\lceil|P_j|/c_R \right\rceil$. 
In this case, we can apply Lemma~\ref{lem:swap} to swap records between $P_{j'}$ and $P_j$. This swap would turn $P_j$ into non-consecutive on the sorted $\CT$.
If this happens, we further invoke Algorithm~\ref{alg:swap} to transform $P_j$ into a consecutive one; moreover, $P_{j'}$ will be put after $P_j$ in the sorted $\CT$. We will repeatedly apply the procedure above until no unordered pair of partitions exists in $\P$.
In fact, a strongly-ordered property also holds by replacing $\left\lceil |P_{j}|/c_R \right\rceil$ with $|P_j|$. We keeps the weakly-ordered version here so that it does not conflict with the following divisible property.

\Paragraph{Divisible Property.} Given a partitioning $\P$ satisfying the consecutive and weakly-ordered property, we next show how to transform it into $\P'$ such that $\P'$ also satisfies divisible property as follows. We start checking the $m$-th partition $P_m$. If $|P_m|$ is divisible by $c_R$, we just skip it. Otherwise, we move records from the previous partition $P_{m-1}$ to $P_m$ until $|P_m|$ is divisible by $c_R$. After done with $P_m$, we move to $P_{m-1}$. We continue this procedure until $|P_m|, |P_{m-1}|, \cdots, |P_2|$ are all divisible by $c_R$.
This fine-grained movement between adjacent partitions naturally preserves the consecutive property as well as the join cost.

\subsubsection{Dynamic Programming}
\label{subsec:dynamic_programming}
With Theorem~\ref{the:main}, we can now reduce the complexity of finding the optimal partitioning in a brute-force way by resorting to dynamic programming, and searching 
candidate partitionings with consecutive, weakly-ordered, and divisible properties. 

\Paragraph{Formulation by Consecutive Property.} We define a sub-problem parameterized by $(i,j)$: finding the optimal partitioning for the first $i$ records in the sorted $\CT$ using $j$ partitions. The join cost of this sub-problem  parameterized by $(i,j)$ is denoted as $V[i][j]$. Recall that the optimal solution of our integer program in Section~\ref{subsec:modeling_io} is exactly $V[n][m]$. As described by Algorithm \ref{alg:matrix_dp}, we can compute $V$ in a bottom-up fashion.
More specifically, we use $V[i][j].\textsf{cost}$ to store the cost of an optimal partitioning for the sub-problem parameterized by $(i,j)$ and $V[i][j].\textsf{LastPar}$ stores the starting position of the last partition correspondingly. 
We also define $\textsf{CalCost}(s,e)$ as the per-partition join $\textsf{PPJ}(R_{j'},S_{j'})$ (omitting the term $\|R_{j'}\|$ in $\textsf{NBJ}(R_{j'}, S_{j'})$ for a pair of partition $R_{j'}$ and $S_{j'}$, if we let $P_{j'} = [s,e]$). 
This can be easily computed using the pre-computed prefix sum:
\begin{equation}
\label{eq:calcost}
\begin{split}
\textsf{CalCost}(s,e)=
\left(\sum_{i=1}^{e}\CT[i] - \sum_{i=1}^{s-1}\CT[i]\right)\cdot & \Bigg\lceil\frac{e - s + 1}{c_R}\Bigg\rceil
\end{split}
\end{equation}
The essence of Algorithm~\ref{alg:matrix_dp} is the following recurrence formula:
\[V[i][j]. \textsf{cost} = \min_{0\le k \le i-1} \left\{V[k][j-1]. \textsf{cost} + \textsf{CalCost}(k+1,i)\right\}\]
which iteratively search the starting position of the last partition in $[0,i-1]$. Also, we can backtrack the index of the last partition to produce the mapping function $f$ from each index to its associated partition, as described by Algorithm~\ref{alg:bt_matrix_cut_pos}.

\begin{algorithm}[h]
\caption{\textsc{Partition}$(\CT, n, m)$}\label{alg:matrix_dp}
Initialize $V$ of sizes $n \times m$\;
\For{$i \gets 1$ to $n$}{
	  $V[i][1].\textsf{cost} \gets \textsf{CalCost}(1,i)$, \ $V[i][1].\textsf{LastPar} \gets 1$\;
}
\For{ $i \gets 2$ to $n$}{
	\For{ $j \gets 2$ to $\min (m, i)$}{
		  $V[i][j].\textsf{cost} \gets +\infty$\;
		\For{ $k \gets 0$ to $i-1$}{
			  $\textsf{tmp} \gets V[k][j-1].\textsf{cost} + \textsf{CalCost}(k+1,i)$\;
			\If{${\normalfont \textsf{tmp}} < V[i][j].\textsf{cost}$}{
				  $V[i][j].\textsf{cost} \gets \textsf{tmp}$\;
				  $V[i][j]. \textsf{LastPar} \gets k+1$\;
                }
                }
		}
  }
\Return $V$\;
\end{algorithm}

\begin{algorithm}[h]
\caption{$\textsc{GetCut}(V, n, m)$}\label{alg:bt_matrix_cut_pos}
$\textsf{LastPar} \gets n$\;
 \For{$j \gets 0$ to $m - 1$}{
    \lFor{$i \gets V[ \normalfont \textsf{LastPar}][m-j].  \normalfont\textsf{LastPar}$ to $ \normalfont\textsf{LastPar}$}{$f(i) \gets m - j$}
    $\textsf{LastPar} \gets V[\textsf{LastPar}][m-j]. \textsf{LastPar} - 1$\;
}
\Return $F$\;
\end{algorithm}

The time complexity of Algorithm~\ref{alg:matrix_dp} is $O(m\cdot n^2)$, which is dominated by the three for-loops, and the space complexity is $O(m\cdot n)$ for storing $V$. This significantly improves the brute-force enumerating method, whose time complexity is as large as $O(m^n)$.

%
%

\Paragraph{Speedup by Weakly-Ordered Property.} Combing the weakly-ordered property in Theorem~\ref{the:main} with the pigeonhole principle, for any sub-problem parameterized by $(i,j)$, there exists an optimal partitioning such that: (\romannumeral 1) the ``largest'' partition contains at least the first $\left\lfloor \frac{i-1}{j}\right\rfloor + 1 - c_R$ records in the sorted $\CT$, and (\romannumeral 2) the ``smallest'' partition contains at most the last $\left\lfloor \frac{i}{j}\right\rfloor + c_R$ records in the sorted $\CT$. Note that the ``largest'' (``smallest'') partition is not essentially the largest (smallest) in terms of the actual partition size.
It is possible that $P_{j'} \precsim P_{j'+1}$ and $\left \lceil |P_{j'}|/c_R\right\rceil =  \left \lceil |P_{j+1}|/c_R\right\rceil$ but $|P_{j'}| \leq |P_{j'+1}|$. However, even if $|P_{j'}| \le |P_{j'+1}|$, the difference cannot be larger than $c_R$ from the weakly-ordered property. 
As such, the ``largest'' partition for the sub-problem parameterized by $(i,j)$ cannot be smaller than $\left\lfloor \frac{i-1}{j}\right\rfloor + 1 - c_R$. Similar analysis can be applied to the upper bound of the size of the ``smallest'' partition. Synthesizing above analysis, we can get two tighter bounds on $k$ in line 8:
\begin{itemize}[leftmargin=*]
    \item {\bf $k \geq \max\{(i - c_R)\cdot(1 - \frac{1}{j}), 0\}$.} The partition $P_j$ contains records from $\CT[k+1:i]$, and $|P_j|$ should be no larger than the last~(``smallest'') partition in $V[k][j-1]$. 
    The last~(``smallest'') partition for the sub-problem parameterized by $(k,j-1)$ is at most $\left\lfloor \frac{k}{j-1} \right\rfloor + c_R$. Hence, $i - k \leq \left\lfloor \frac{k}{j-1} \right\rfloor+c_R$. 
    \item {\bf $k \leq \max\{i - \left\lfloor \frac{n-i-1}{m-j}\right\rfloor - 1 + c_R,1\}$.} The partition $P_j$ contains records from $\CT[k+1:i]$, and it should be no smaller than the ``largest'' partition of remaining records in $\CT[i+1:n]$. Recall that, we still need to put the remaining $n-i$ records into $m-j$ partitions. The largest partition of the sub-problem parameterized by $(n-i,m-j)$ contains at least $\left\lfloor \frac{n-i-1}{m-j}\right\rfloor + 1 - c_R$ records. Hence, $i - k \geq \left\lfloor \frac{n-i-1}{m-j}\right\rfloor + 1 - c_R$.
\end{itemize}

After changing the range of the third loop (line 8,  Algorithm~\ref{alg:matrix_dp}) to $[\max\{(i - c_R)\cdot(1 - \frac{1}{j}), 0\}, \max\{i - \left\lfloor \frac{n-i-1}{m-j}\right\rfloor - 1 + c_R,1\}]$,  
we can skip the calculation of $V[i][j]$ when $(i - c_R)\cdot(1 - \frac{1}{j}) > i - \left\lfloor \frac{n-i-1}{m-j}\right\rfloor - 1 + c_R$. This pruning reduces the time complexity to $O(n^2 \log m)$.

\Paragraph{Speedup by Divisible Property.} 
%
To utilize the divisible property, we initially put the first $(n\mod c_R)$ records in the first partition, and then change the step size (line 5 and line 8) of Algorithm~\ref{alg:matrix_dp} to $c_R$. This way, we can shrink the size of $V$ to $\left(\left\lceil n/c_R\right\rceil + 1\right)\times m$. Recall that $m$ can be as large as $B-1$, $F$ is a constant larger than 1, and thus $c_R = \left\lfloor b_R\cdot (B-2) / F \right\rfloor \ge m$. Therefore, the time complexity can be reduced to $O\left(\left(n/c_R \right)^2 \log m \right) = O\left(\frac{n^2 \log m}{m^2}\right)$ and the space complexity to $O\left(n/c_R \cdot m\right)=O(n)$.

\begin{figure}
\centering
     \begin{subfigure}{0.3\linewidth}
      \includegraphics[width=\linewidth]{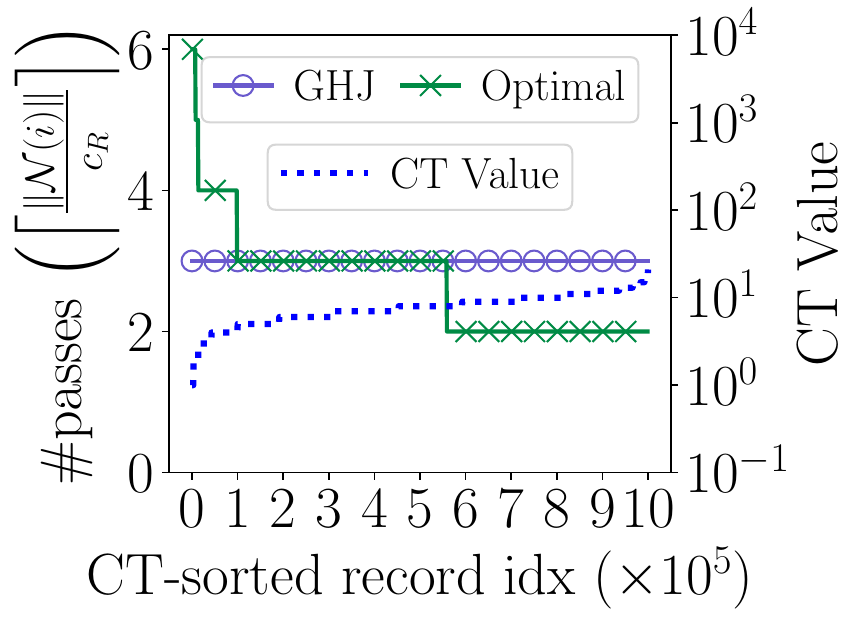}
      \caption{Uniform}
         \label{fig:case-study-uni}
     \end{subfigure}
     \begin{subfigure}{0.3\linewidth}
        \includegraphics[width=\linewidth]{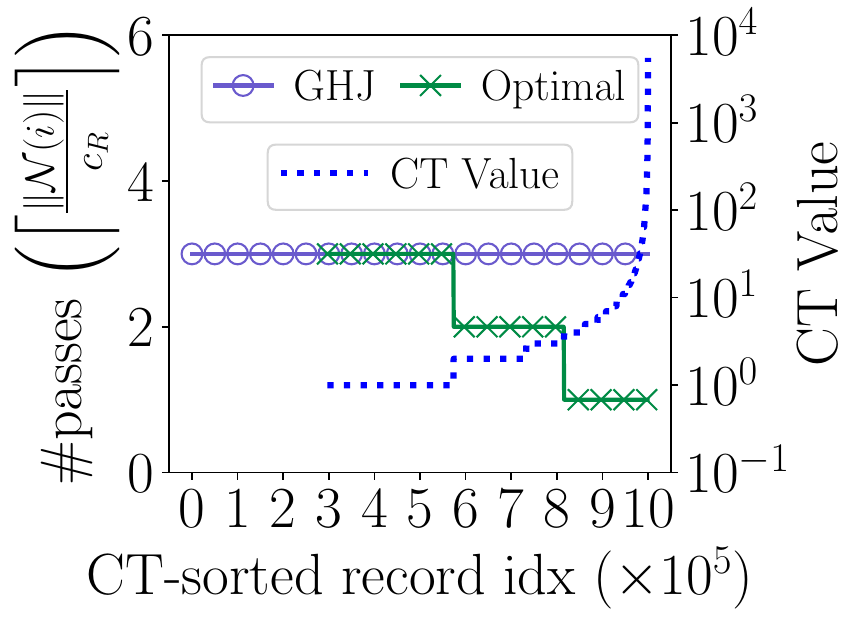}
         \caption{Zipfian}
         \label{fig:case-study-zipf}
     \end{subfigure}
     \vspace{-0.1in}
     \caption{A comparison of the number of passes required for the outer relation $S$ between the partitioning adopted by \textsf{DHH}/\textsf{GHJ} and our optimal partitioning,  assuming $B \le \sqrt{\|R\|\cdot F}$ (\textsf{DHH} without skew optimization downgrades to \textsf{GHJ}), $n_R=1M$, $n_S = 8M$, $F=1.02$ and $B=320$ pages. Moreover, the record size of $R$ is $1KB$, thus $R$ occupies 250K pages. In (a), random (uniform) partitioning (\textsf{GHJ}), the partition size is around $250K/319 \approx 784$ pages which require $\lceil 783/(318\times F) \rceil = 3$ passes to scan every partition of $S$. In contrast, the optimal partitioning allows the number of passes to vary from 2 to 6, which follows a non-monotonically decreasing trend as \CT{} value increases. When we have a skewed correlation (b), a similar pattern is observed. Further, the optimal partitioning naturally excludes records with $\CT[i]=0$.}
\label{fig:case-study}
\end{figure}

\subsection{Optimal Partitioning With Caching}
\label{subsec:hybrid-join}
We move to the general case when records can be cached. In this case, some records stay in memory during the partition phase, so as to avoid repetitive reads/writes, while remaining records still go through the partitioning phase as Section~\ref{subsec:opt_partition_no_caching}. A natural question arises: which records should stay in memory and which records should go to partition?
As we can see in DHH (Figure~\ref{fig:dynamic_hybrid_join_intro}), the hash function $h_{\textsf{split}}$, together with POB actually categorizes all the join keys by their indexes into two sets, memory-resident keys as $K_{\textsf{mem}} \subseteq [n]$ and disk-resident keys as $K_{\textsf{disk}} = [n] - K_{\textsf{mem}}$.
Then, it suffices to find an optimal partitioning (without caching) for $K_{\textsf{disk}}$, which degenerates to our problem in Section~\ref{subsec:opt_partition_no_caching}. Together, we derive a general join cost function $\JC(K_{\textsf{disk}}, \P, m)$ as:
\[ \left\lceil\sum\limits_{i \in K_{\textsf{disk}}} \frac{\CT[i]}{b_S}  \cdot \lceil \frac{|\mathcal{N}_f(i)|}{c_R} \rceil \right\rceil  + \left(1+\mu\right)\cdot \left\lceil \frac{|K_{\textsf{disk}}|}{b_R} \right\rceil +  \mu \cdot  \sum\limits_{i \in K_{\textsf{disk}}} \left\lceil\frac{\CT[i]}{b_S} \right\rceil\]
where $b_R$ (resp. $b_S$) represents the number of records per page for relation $R$ (resp. $S$), and $\mu$ is the ratio between random write and sequential read.
Compared to the original cost function $\JC(\P,m)$, the new cost function $\JC(K_{\textsf{disk}}, \P,m)$ involves two additional terms -- $\mu \cdot (\left\lceil |K_{\textsf{disk}}|/{b_R} \right\rceil + \left\lceil \sum_{i \in K_{\textsf{disk}}}\CT[i]/{b_S} \right\rceil)$ as the cost of partitioning both relations, and $\left\lceil |K_{\textsf{disk}}|/{b_R} \right\rceil$ as the cost of loading $R$ in the probing phase. Together, we have the following integer program:
\begin{equation*}
\begin{split}
\min_{K_{\textsf{disk}}, \P} & \quad \JC(K_{\textsf{disk}}, \P, m)\\
\textrm{subject to}
& \quad \left\lceil \frac{n - |K_{\textsf{disk}}|}{b_R}\cdot F  \right\rceil + m + 2\leq B,\\
& \quad \sum\limits_{j=1}^{m}\P_{i,j}=1, \forall i \in [n] \\
& \quad \P_{i,j} \in \{0,1\}, \forall i \in [n], j \in [m]
\end{split}
\label{eq:JC_def_hybrid}
\end{equation*}
where the first constraint limits the size of the memory-resident hash table, as $2$ pages are reserved for input and join output, and $m$ pages for the output buffer for each disk-resident partition. 
\begin{algorithm}[t]
\caption{\textsc{\ExactAlgName{}}$(\CT, n, B)$}\label{alg:hybrid_matrixdp}
  $\textsf{tmp} \gets +\infty$, $f \gets \emptyset$\;
\For{$k \gets 0$ to $c_R$}{
	$V \gets  \textbf{\textsc{Partition}}(\CT[1:n-k], n - k, B - 2 - \lceil \frac{k\cdot F}{b_R} \rceil)$\;
	  $c_{\textsf{probe}} \gets \lceil \frac{n-k}{b_R} \rceil + V[n-k][B - 2 - \lceil \frac{k\cdot F}{b_R} \rceil].\textsf{cost} $\;
	  $c_{\textsf{part}} \gets \mu \cdot (\lceil \frac{n-k}{b_R} \rceil + \lceil \frac{1}{b_S} \cdot \sum_{i=1}^{n-k}\CT[i] \rceil )$\;
	\If{{\normalfont $\textsf{tmp} < c_{\textsf{probe}} + c_{\textsf{part}}$}}{
		  $\textsf{tmp} \gets c_{\textsf{probe}} + c_{\textsf{part}}$\;
		  $f \gets \textsc{GetCut}(V,n-k,B - 2 - \lceil \frac{k\cdot F}{b_R} \rceil)$\;
	}
}
  \Return $f$\;
\end{algorithm}

The integer program above can be interpreted by first fixing the $K_{\textsf{disk}}$  (i.e., which keys/records should be spilled out to disk), and then finding the optimal partitioning for $K_{\textsf{disk}}$.
Our Algorithm \ExactAlgName{} can be invoked as a primitive for given $K_{\textsf{disk}}$.
It remains to identify the optimal $K_{\textsf{disk}}$ (i.e., which records should be spilled out to disk) that leads to the overall minimum cost. 
Although there is an exponentially large number of candidates for $K_{\textsf{disk}}$, it is easy to see that records with low $\CT$ value should be spilled out to disk, while records with high $\CT$ values should be kept in memory.
In other words, the consecutive property of Theorem~\ref{the:main} still applies for this hybrid partitioning -- we can achieve the optimal cost by caching the top-$k$ records from $\CT$ in memory, if restricting the size of $K_{\textsf{mem}}$ to be $k$, where $k < c_R$ since at most $c_R$ records can remain in memory.
By extending Algorithm~\ref{alg:matrix_dp} to a hybrid version, we finally come to the 
\ExactAlgName{} in Algorithm~\ref{alg:hybrid_matrixdp}.

\Paragraph{Complexity Analysis.} As the number of records that can remain in memory is at most $c_R=\lfloor b_R\cdot (B-2)/F\rfloor$, we run Algorithm~\ref{alg:matrix_dp} at most $c_R$ times. 
Recall that the complexity of Algorithm~\ref{alg:matrix_dp} can be reduced to $O(n^2\log m/c_R^2)$, the time complexity for \ExactAlgName{} is now $O(n^2\log m/c_R) =O(n^2\log m/m)$. The space complexity is $O(n)$.

%% file: 4-practical-join-algorithm.tex
\section{Our Practical Algorithm}
\label{sec:practical-partitioning-algorithm}
In practice, it is impossible to know the exact correlation between input relations for each key in advance; instead, a much smaller set of high-frequency keys are collected for skew optimization, such as the Most Common Values (MCVs) in PostgreSQL~\cite{PostgreSQLa}. PostgreSQL supports skew optimization by assigning 2\% memory to build a hash table for skewed keys if their total frequency is larger than 1\% of the outer relation, but this heuristic uses fixed thresholds to trigger skew optimization, which may be sub-optimal for arbitrary correlation skewness.
Hence, we present a practical algorithm built upon the theoretically I/O-optimal algorithm in the previous section.
Instead of relying on the whole correlation $\CT$ table, our practical algorithm only needs the same amount of MCVs as PostgreSQL and generates a hybrid partitioning that can be adaptive to arbitrary correlation skewness, and constrained memory budget.

\subsection{Hybrid Partitioning}
\label{subsec:approximate_algorithm}
Let $K$ be the set of skew keys tracked from MCVs with their \CT~ values known.
Guided by the consecutive property and weakly-ordered property in Theorem~\ref{the:main}, it is always efficient to keep as many skew keys as possible in memory if possible, so that we can avoid reading the corresponding records in $S$ repeatedly. For the rest of keys not tracked by MCVs, we consider their frequency as low and can be tackled well by the traditional hashing method \textsf{GHJ}/\textsf{DHH}.

\Paragraph{Framework.} In a hybrid partitioning, we design a hash set $HS_{\textsf{mem}}$ to store a subset of skew key (denoted as $K_{\textsf{mem}} \subseteq K$) that will be used to build the in-memory hash table $HT_{\textsf{mem}}$, a hash map $f_{\textsf{disk}}$ to store a subset of skew keys (denoted as $K_{\textsf{disk}} \subseteq K$) with their designated partition identifiers, and the rest of skew keys (if there exists any) together with non-skew keys (denoted as $K_{\textsf{rest}}$) will be uniformly partitioned by \textsf{GHJ}/\textsf{DHH} with a certain memory budget $m_{\textsf{rest}}$ (in pages). With $HS_{\textsf{mem}}$, $f_{\textsf{disk}}$ and $m_{\textsf{rest}}$ in hand, we can partition records in $R$ using Algorithm~\ref{alg:three_set_partition_R}.
For each record in $R$, if its key is in $HS_{\textsf{mem}}$, we put the entire record in the in-memory hash table $HT_{\textsf{mem}}$ (the hash set $HS_{\textsf{mem}}$ only stores the key);
if its key belongs to $f_{\textsf{disk}}$, we put it into the partition specified by $f_{\textsf{disk}}$, which will be spilled to disk later;
otherwise, it will be assigned by \textsf{GHJ}/\textsf{DHH} along with the rest of the records using the $m_{\textsf{rest}}$ pages reserved. We can partition records in $S$ similarly using Algorithm~\ref{alg:three_set_partition_S}. This hybrid partitioning workflow is visualized in Figure~\ref{fig:NOCAP_Join}.

\begin{figure}[t]
\includegraphics[width=0.75\textwidth]{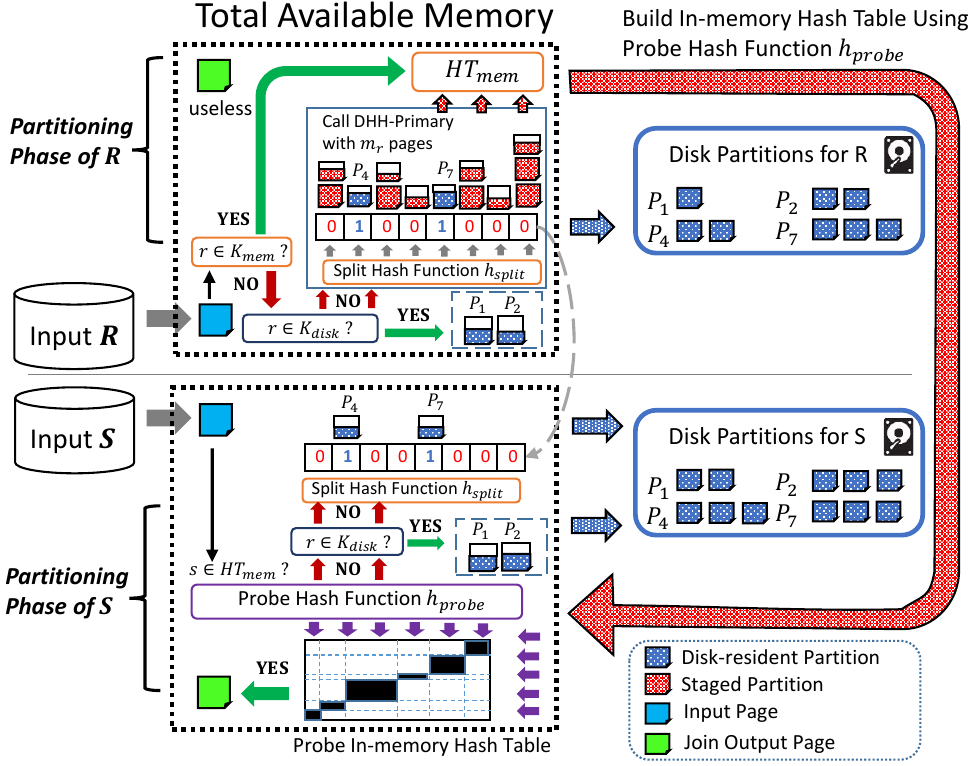}
\vspace{-0.1in}
\caption{An illustration of the hybrid partitioning workflow of NOCAP. Partitions $P_4$ and $P_7$ are spilled to disk when invoking \textsf{DHH}, 
while $P_1$ and $P_2$ are written to disk due to designated partitioning by $f_{disk}$.
Other records in $R$ either have keys from $K_{\textsf{mem}}$ or stage in memory when invoking $\textsf{DHH}$, which will be inserted to the in-memory hash table $HT_{\textsf{mem}}$.}
\label{fig:NOCAP_Join}
\end{figure} 

\begin{figure}[t]
\includegraphics[width=0.65\textwidth]{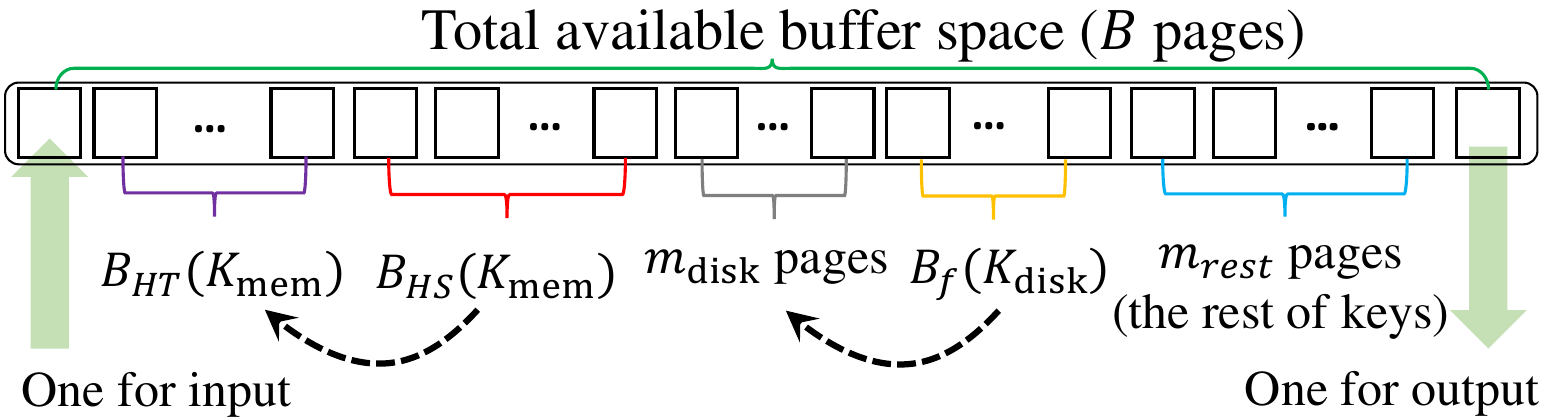}
\vspace{-0.1in}
\caption{The memory breakdown when partitioning $R$. 
}
\label{fig:memory_division_in_approximate_algorithm}
\end{figure} 

\begin{algorithm}[h]
\caption{$\textsc{Hybrid-Partition-Primary}(R, HS_{\textsf{mem}}, f_{\textsf{disk}}, m_\textsf{rest})$}\label{alg:three_set_partition_R}
$R' \gets \{r \in R: r.\textsf{key} \notin HS_{\textsf{mem}},r.\textsf{key} \notin f_{\textsf{disk}}\}$;\Comment{Logically construct $R'$ without instantiation}\;
\ForEach{$r \in R-R'$}{
    \lIf{{\normalfont $r.\textsf{key} \in HS_{\textsf{mem}}$}}{add $r$ to $HT_{\textsf{mem}}[h_\textsf{probe}(r.\textsf{key})]$}
    \lElse{assign $r$ to disk-resident
            partition $f_{\textsf{disk}}(r.\textsf{key})$} 
}
$(HT'_{\textsf{mem}}, \POB) \gets$ \textsc{DHH-Primary}$(R', m_{\textsf{rest}})$\;
\Return $(HT'_{\textsf{mem}} \cup HT_{\textsf{mem}}, \POB)$
\end{algorithm}
\begin{algorithm}[h]
\caption{$\textsc{Hybrid-Partition-Foreign}(S, HT_{\textsf{mem}}, f_{\textsf{disk}},\POB)$}\label{alg:three_set_partition_S}
$S' \gets \{s \in S: s.\textsf{key} \notin f_{\textsf{disk}}\}$;\Comment{Logically construct $S'$ without instantiation}\;
\ForEach{$s \in S'$}{ 
    \If{\normalfont $s.\textsf{key} \in f_{disk}$}{
          assign $s$ to disk-resident
            partition $f_{\textsf{disk}}(s.\textsf{key})$
        }
}
\textsc{DHH-Foreign}$(S - S', HT_{\textsf{mem}}, \POB)$\;
\end{algorithm}

Regarding whether to use \textsf{GHJ} or \textsf{DHH}, we note the following difference. When the available memory is large, say $\lceil |K_\textsf{rest}|/c_R \rceil < m_{\textsf{rest}}$,  we simply invoke \textsf{DHH} to partition the rest of keys $K_\textsf{rest}$ using $m_{\textsf{rest}}$ pages. When the available memory is small, as discussed in Section~\ref{subsec:dynamic_hhj}, the \textsf{DHH} just downgrades to the \textsf{GHJ}.

\Paragraph{Enforcing Memory Constraints.} State-of-the-art systems (e.g., MySQL and PostgreSQL) typically assign to each join operator a user-defined memory budget.
The default memory limit for the join operator in MySQL is 256 KB~\cite{MySQL2021} and the one in PostgreSQL is 4 MB~\cite{PostgreSQL2022}.
Assuming the available buffer is in total $B$ pages, we need to ensure that Algorithm~\ref{alg:three_set_partition_R} uses no more than $B$ pages.
We present a memory breakdown as illustrated in Figure~\ref{fig:memory_division_in_approximate_algorithm}.
\begin{itemize}[leftmargin=*]
    \item $B_{HS}(K_{\textsf{mem}})$:\#pages for the hash set $HS_{\textsf{mem}}$ with keys in $K_{\textsf{mem}}$;
    \item $B_{HT}(K_{\textsf{mem}})$:\#pages for the hash table $HT_{\textsf{mem}}$ with records whose keys are in $K_{\textsf{mem}}$;
    \item $B_f(K_{\textsf{disk}})$: \#pages for the hash map $f_\textsf{disk}$ with keys in $K_{\textsf{disk}}$;
    \item $m_\textsf{disk}$: \#pages as write buffer for all the on-disk partitions specified in $f_\textsf{disk}$ (which is the same as \#partitions specified in $f_\textsf{disk}$);
    \item $m_{\textsf{rest}}$: \#pages to partition records whose keys are in $K_{\textsf{rest}}$;
\end{itemize}
Moreover, we denote $ks$ as the size of each key (in bytes), and $ps$ as the size of each page (in bytes). We can roughly have $B_{HT}(K_{\textsf{mem}})=\lceil b_R\cdot |K_{\textsf{mem}}|/F\rceil$, $B_{HS}(|K_{\textsf{mem}}|)=\lceil ks\cdot |K_{\textsf{mem}}|/(F\cdot ps)\rceil$, and $B_{f}(|K_{\textsf{disk}}|) =\lceil (ks + 4)\cdot |K_{\textsf{disk}}|/(F\cdot ps)\rceil$, where a 4-byte integer is used to store the partition identifier in the hash map $f_{\textsf{disk}}$.  
At last, we come to the memory constraint subject to $B$ (the total number of pages):
\begin{equation*}
B_{HS}(|K_{\textsf{mem}}|) + B_{HT}(|K_{\textsf{mem}}|) + B_f(|K_{\textsf{disk}}|) + m_{\textsf{disk}} + m_\textsf{rest}\leq B - 2
\label{eq:formal_mem_constraint}
\end{equation*}

\Paragraph{Cost function.} Combining all these above together, we come to the cost function and integer programming as follows:
\begin{align*}
    \min_{K_\textsf{mem}, K_\textsf{disk}, \P_\textsf{disk}} \ \ & \ \JC(K_\textsf{disk}, \P_\textsf{disk}, m_\textsf{disk}) + g_{\textsf{DHH}}(|K_{\textsf{disk}}|, m_{\textsf{rest}}) \\
   \textrm{s.t.} \ \ 
    & \ B_{HS}(|K_{\textsf{mem}}|) + B_{HT}(|K_{\textsf{mem}}|) 
    + B_f(|K_{\textsf{disk}}|) \leq B - 2 -m_{\textsf{disk}} - m_\textsf{rest}, \\
    & 0 \le m_{\textsf{disk}}, 0\le  m_{\textsf{rest}} \leq B - 2, \\
    & K_\textsf{mem} \cap K_\textsf{disk} = \emptyset, K_{\textsf{mem}} \cup K_{\textsf{disk}} \subseteq K \\
    & \sum\limits_{j=1}^{m}\P_{i,j}=1, \forall i \in [n] \\
    & \P_{i,j} \in \{0,1\}, \forall i \in [n], \forall j \in [m]
\end{align*}

In the objective function, the first term $\JC(K_\textsf{disk}, f_\textsf{disk}, m_\textsf{disk})$ is the cost of partitioning $K_\textsf{disk}$ into $m_\textsf{disk}$ partitions, and the second term $g_{\textsf{DHH}}(|K_{\textsf{rest}}|, m_{\textsf{rest}})$ is the estimated cost of partitioning $K_\textsf{rest}$ with $m_{\textsf{rest}}$ pages using either \textsf{GHJ} or \textsf{DHH}, depending on how large $m_{\textsf{rest}}$ is, as mentioned earlier.
We will discuss $g_{\textsf{DHH}}(|K_{\textsf{disk}}|, m_{\textsf{rest}})$ in Section~\ref{subsec:rounded_hash_join}. In both estimation models of \textsf{GHJ} or \textsf{DHH}, we need to know the total number of records from $S$ whose key falls into $K_\textsf{rest}$, which can be estimated by $n_S - \sum_{i \in K_\textsf{mem} \cup K_\textsf{disk}} \CT[i]$, since the \CT~values of keys in $K$ are known from MCVs. 

To find out the best combination of $K_{\textsf{mem}}$, $K_{\textsf{rest}}$ and $f_{\textsf{disk}}$ (the optimal value of $m_\textsf{rest}$ can be calculated using Equation~\eqref{eq:formal_mem_constraint}), we iterate all possible combinations of $|K_{\textsf{mem}}|$, $|K_{\textsf{disk}}|$ and $f_{\textsf{disk}}$, and pick the one with the minimum estimated cost, as described by Algorithm~\ref{alg:find_best_Mk}. 
Then, implied by the consecutive property and weakly-ordered property in Theorem~\ref{the:main} (i.e., keys with higher frequency should be kept in memory or in a small partition with higher priority), we always pick the top-$|K_{\textsf{mem}}|$ keys from $K$ as $K_{\textsf{mem}}$ and the next top-$|K_{\textsf{disk}}|$ keys from $K$ as $K_{\textsf{disk}}$.

\begin{algorithm}[h]
\caption{$\textsc{NOCAP}(\CT, k, n, m)$}\label{alg:find_best_Mk}
$c_{opt} \gets +\infty$, $k_{\textsf{mem}} \gets 0$, $k_{\textsf{disk}} \gets 0$, $f_{\textsf{disk}} \gets \emptyset$\;
\For{$i_1 \gets 0$ to $\min(k,c_R)$}{
	$V \gets \mathbf{\textsc{Partition}}(\CT[i_1+1:k], k-i_1, \left\lceil \frac{k-i_1}{c_R} \right\rceil)$\;
	\For{$i_2 \gets 0$ to $\min(k,c_R) - i_1$} {
		\For{$j \gets \min(i_2,1)$ to $\lceil \frac{i_2}{c_R} \rceil$} {
			$m_{\textsf{rest}} \gets B-2-B_{HT}(i_1)-B_{HS}(i_1)-B_{f}(i_2) - j$\;
			$c_{\textsf{probe}} \gets V[i_2][j].\textsf{cost}$\;
			$c_{\textsf{part}} \gets \mu \cdot \left(\lceil \frac{i_2}{b_R} \rceil + \lceil \frac{1}{b_S} \cdot \sum_{j=i_1+1}^{j=i_2}\CT[j] \rceil \right)$\;
			$c_{\textsf{\textsf{rest}}} \gets g_{\textsf{DHH}}(i_2+1, n, m_{\textsf{rest}})$\;
			$c_{\textsf{tmp}} \gets c_{\textsf{probe}}+c_{\textsf{part}}+c_{\textsf{rest}}$\;
			\If{{\normalfont $c_{\textsf{tmp}} < c_{\textsf{opt}}$}} {
			$c_{opt} \gets c_{\textsf{tmp}}$, $k_{\textsf{mem}}\gets i_1,k_{\textsf{disk}} \gets i_2$\;
			$f_{\textsf{disk}} \gets \mathbf{\textsc{GetCut}}(V,i_2,j) $\;
		}
		}
	}
}
\Return $k_{\textsf{mem}}$, $k_{\textsf{disk}}$, $f_{\textsf{disk}}$\;
\end{algorithm}

\subsection{Optimizations with Rounded Hash}
\label{subsec:rounded_hash_join}
In our implementation, we also design {\em rounded hash} (RH) in partitioning that can further reduce the I/O cost of scanning the outer relation in the per-partition join. This is motivated by the divisible property in Theorem~\ref{the:main}, that the sizes of most partitions should be divisible by $c_R$ in the optimal partitioning.
To understand why this insight may help, let us consider an example in Figure~\ref{fig:uniform-vs-non-uniform-partition}. 

We note that uniform partitioning namely assigns records using a uniform hash function $\textsf{PartID}= \textsf{hash}(\textsf{key})\mod m$, based on the modular function with respect to $m$ (the number of partitions), which is denoted as {\em plain hash} (PH). To reduce the extra I/O cost using non-uniform partitioning, our RH uses an additional modular function with respect to $\left\lceil n/c_R \right\rceil$:
\begin{equation}
\textsf{PartID}= \left( \textsf{hash}(\textsf{key})\mod \left\lceil n/c_R \right\rceil\right)\mod m
\label{eq:rounded-partition}
\end{equation}
where $c_R = \left\lfloor b_R\cdot (B-2) / F \right\rfloor$ is the chunk size of records. Intuitively, our RH maximizes the number of partitions with $\left\lfloor \left\lceil n/{c_R}  \right\rceil /(B-1) \right\rfloor$ chunks. In addition to reducing unnecessary passes, RH also merge small partitions to avoid fragmentation when partitioning~\cite{Kitsuregawa1989}.

\begin{figure}[t]   
        \includegraphics[width=0.65\linewidth]{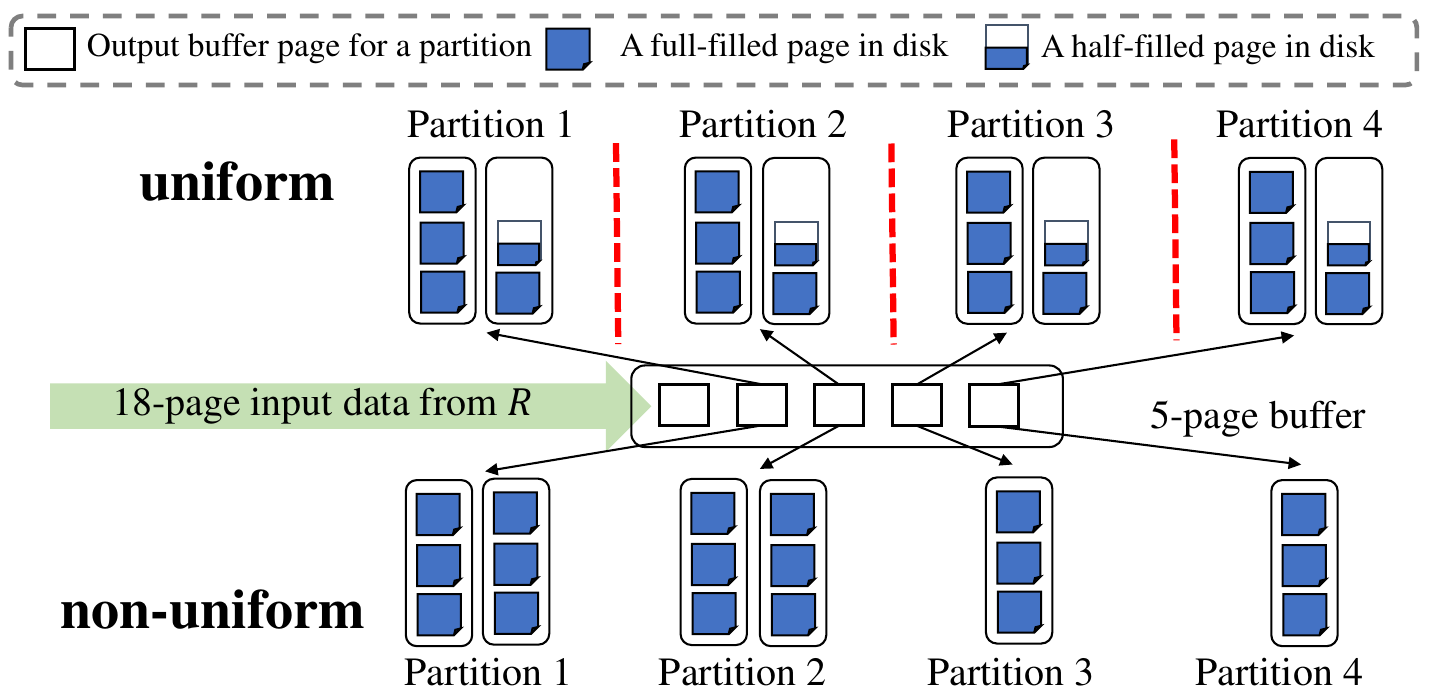}
        \vspace{-0.1in}
          \caption{An illustration of the uniform and non-uniform partitioning on the inner relation $R$, with a 5-page buffer (with 4 pages for partitioning $R$ and 1 page for streaming the input). In \textsf{NBJ}, the maximum \#chunks loaded at a time is 3-page long (with 1 page for streaming $S$ and $1$ page for the output assuming $F=1.0$). Uniform partitioning leads to 4 partitions with each having $18/4=4.5$ pages worth of data. One needs to read $S$ twice for each partition. As a comparison, rounded hashing with $n/{c_R} =6$ and $m=4$ generates 2 partitions with each having 6 pages, and 2 partitions with each having 3 page worth of data. In that way, partitions 1 and 2 need two passes on the corresponding partitions in $S$, while partitions 3 and 4 only need a single pass.
          }
          \centering
         \label{fig:uniform-vs-non-uniform-partition}
\end{figure} 

\Paragraph{Parametric Optimization.} Since we rely on a hash function to partition records, partitions 1 and 2 in Figure~\ref{fig:uniform-vs-non-uniform-partition} could occasionally ``overflow'' and grow larger than six pages due to randomness.
In a skewed workload, overflowed partitions could result in higher cost compared to the cost of uniform partitioning.
In fact, a more robust scheme in Figure~\ref{fig:uniform-vs-non-uniform-partition} is to assign 2.5 pages to partition 4 and evenly distribute 15.5 pages to partitions 1, 2, and 3.
Formally, we replace $c_R$ with $c_R^*=\beta \cdot c_R$ in Equation~\eqref{eq:rounded-partition}, where $\beta$ should be very close to 1 in the range $(0,1]$ (we fix $\beta=0.95$ in our implementation). 

\Paragraph{Cost Estimation.} To get a rough estimation of how many I/Os are used by RH, we first build the I/O model for PH-based partitioning. 
We actually consider a generalized problem by estimating the per-partition cost for partitioning records lying on the \CT\ array from index $s$ to index $e$ into $m$ partitions using PH, noted by $g_{PH}(s,e,m)$:
\begin{equation*}
g_{\textsf{PH}}(s,e,m)= 
\sum_{j=1}^m \mathbb{E}\left[ \left\lceil \frac{|P_j|}{c_R} \right\rceil \right] \cdot \sum\limits_{i \in P_j} \CT[i] 
 = \left\lceil \frac{e - s + 1}{m \cdot c_R} \right\rceil \cdot \sum\limits_{i = s}^{e}\CT[i] 
\end{equation*}
where each $\|P_j\|$ is approximated as a Poisson distribution with $\lambda = \frac{e-s+1}{m}$. 
When estimating the cost of RH we need to know how much proportion (noted by $\gamma$) of data from $S$ is scanned with one fewer pass ($1-\gamma$ of data is scanned with one more pass):
$$
\gamma = \left(\left(\left\lceil \frac{e - s + 1}{c_R^*}\right\rceil \text{ mod } m \right)\cdot \left\lfloor \frac{e - s + 1}{m\cdot c_R^*} \right\rfloor \cdot c_R^*\right)/(e - s + 1)
$$
As such, the normalized number of rounded passes is:
\begin{equation}
\#\textsf{rounded}\_\textsf{passes}(s,e) = \gamma \cdot \left\lfloor \frac{e - s + 1}{m\cdot c_R^*} \right\rfloor + (1-\gamma) \cdot \left\lceil \frac{e - s + 1}{m\cdot c_R^*} \right\rceil
\nonumber
\end{equation}
We then have our cost model for rounded hash $g_{\textsf{RH}}$:
\begin{equation}
g_{\textsf{RH}}(s, e,m)=\#\textsf{rounded}\_\textsf{passes}(s,e) \cdot  \sum\limits_{i = s}^{e}\CT[i] 
\label{eq:expected_rounded_hash_join_cost}
\end{equation}
The cost of our original problem is captured by $g_{\textsf{RH}}(1,n,m)$.

\Paragraph{Overestimation using Chernoff Bound.}
When the filling percent for each partition using PH reaches the predefined threshold $\beta$, ($\frac{e-s+1}{m} > \beta\cdot t \cdot c_R$, where $t$ is the largest positive integer such that $t\cdot c_R > \frac{e-s+1}{m} $), we disable RH. 
This is to avoid additional passes that are triggered by occasional overflow in the random process.
We also need to take into account the possible overflow in the cost model.
In fact, every element is distributed independently and identically using the same hash function, and thus we can apply Chernoff bound to overestimate the probability.
For an arbitrary partition, we define $X=\sum_{i=s}^{e}X_i$ where $X_i=1$ (i.e., $i$-th element is assigned to this partition) with probability $\frac{1}{m}$, and $X_i=0$ with probability $1-\frac{1}{m}$.
Then $\mathbb{E}[X]=\frac{e-s+1}{m}$. From Chernoff bound, $\Pr[X > t\cdot c_R] < \left( \frac{e^\sigma}{(1+\sigma)^{1+\sigma}}\right)^{\mathbb{E}[X]}$,
where $\sigma = \frac{t\cdot c_R\cdot m}{e-s+1} - 1$. We thus let $1-\gamma =\left( \frac{e^\sigma}{(1+\sigma)^{1+\sigma}}\right)^{\mathbb{E}[X]}$ and set:
$$
\#\textsf{rounded}\_\textsf{passes}(s,e) =
\gamma \cdot \Bigg\lceil \frac{e - s + 1}{m\cdot c_R^*} \Bigg\rceil + (1-\gamma)\cdot \left( \Bigg\lceil \frac{e - s + 1}{m\cdot c_R^*} \Bigg\rceil + 1\right) 
$$
We reuse Equation~\eqref{eq:expected_rounded_hash_join_cost} to estimate the I/O cost.

%% file: 5-exp.tex
\section{Experimental Analysis}
\label{sec:exp}

\begin{figure*}[t!]
\centering
     \begin{subfigure}{0.24\linewidth}
         \centering
         \includegraphics[width=\linewidth]{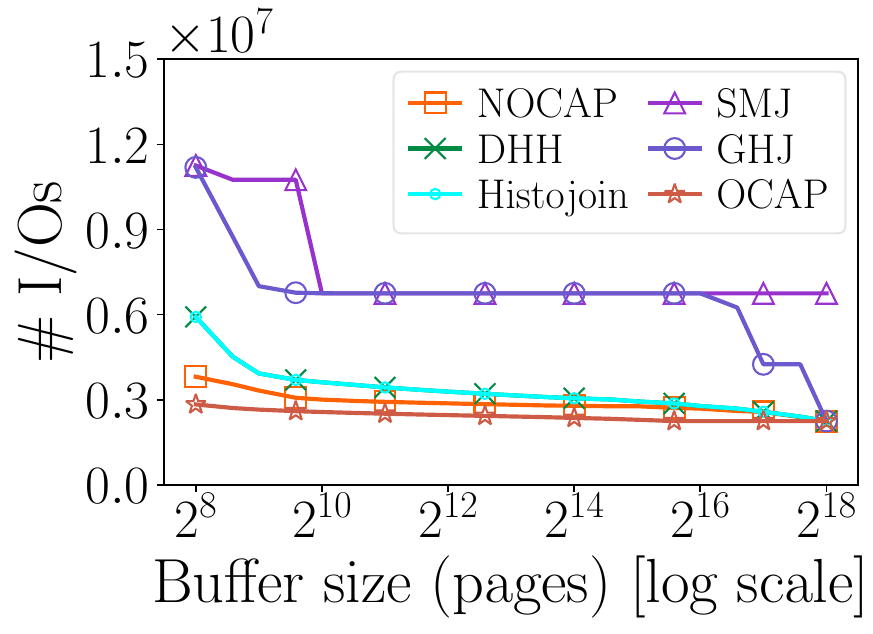}
         \caption{Zipf ($\alpha$=$1.3$)}
         \label{fig:varying-buffer-zipf-1.3-io}
     \end{subfigure}
     \hfill
      \begin{subfigure}{0.24\linewidth}
         \centering
         \includegraphics[width=\linewidth]{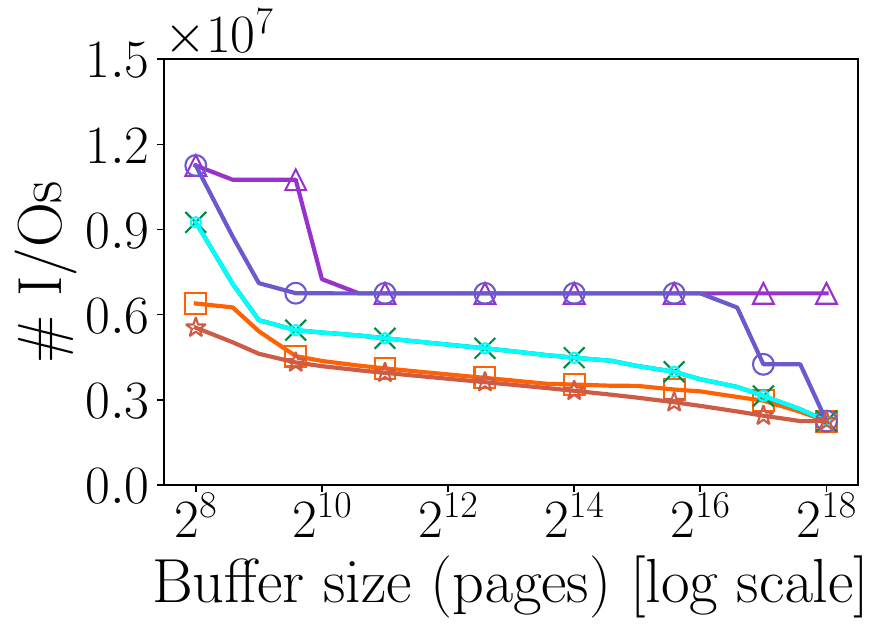}
         \caption{Zipf ($\alpha$=$1.0$)}
         \label{fig:varying-buffer-zipf-1.0-io}
     \end{subfigure}
     \hfill
  	\begin{subfigure}{0.24\linewidth}
         \centering
         \includegraphics[width=\linewidth]{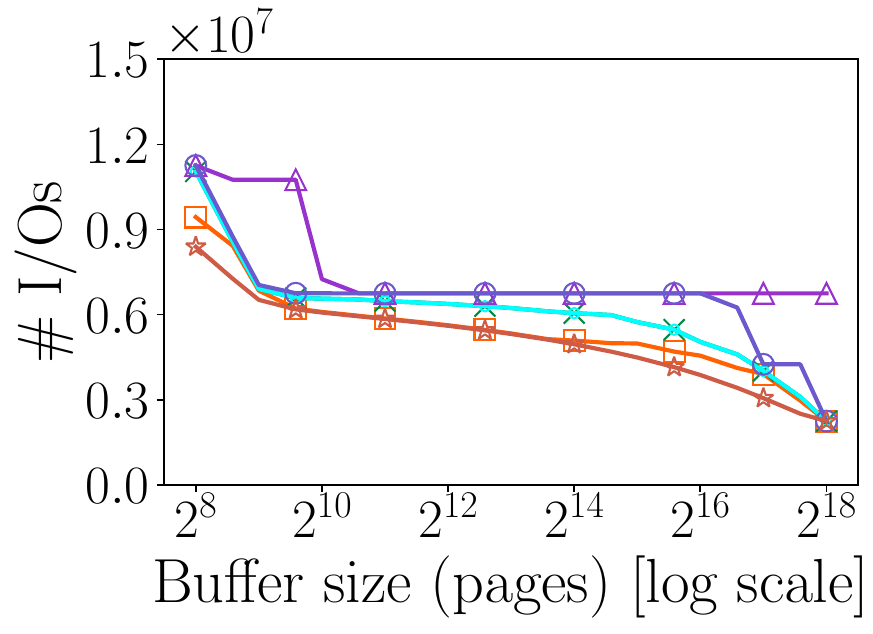}
         \caption{Zipf ($\alpha$=$0.7$)}
         \label{fig:varying-buffer-zipf-0.7-io}
     \end{subfigure}
     \hfill
     \begin{subfigure}{0.24\linewidth}
         \centering
         \includegraphics[width=\linewidth]{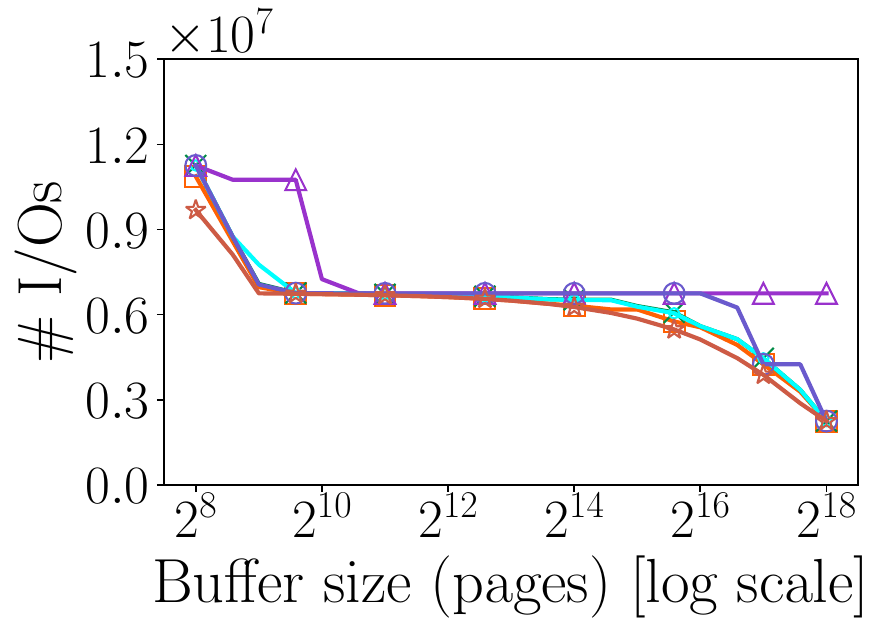}
         
         \caption{Uniform}
         \label{fig:varying-buffer-uni-io}
     \end{subfigure}
     \hfill
     \begin{subfigure}{0.24\linewidth}
         \centering
         \includegraphics[width=\linewidth]{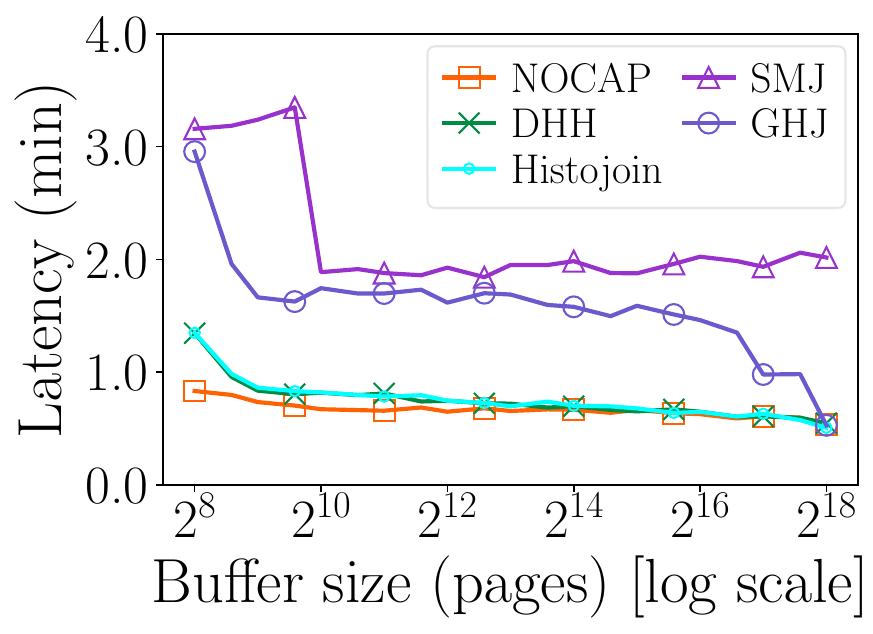}
         \caption{Zipf ($\alpha$=$1.3$) [w/o sync]}
         \label{fig:varying-buffer-zipf-1.3-lat}
         
     \end{subfigure}
     \hfill
      \begin{subfigure}{0.24\linewidth}
    
         \centering
         \includegraphics[width=\linewidth]{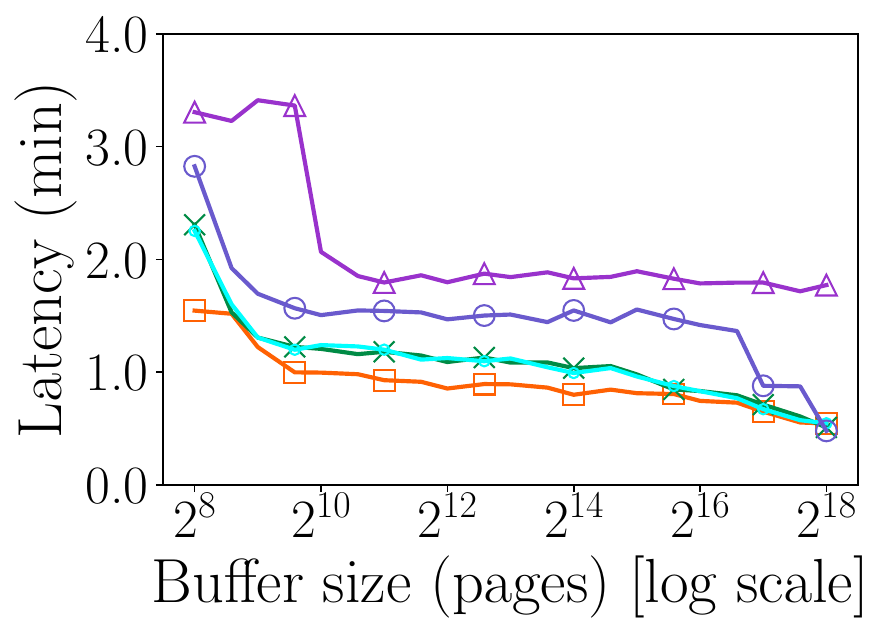}
        
         \caption{Zipf ($\alpha$=$1.0$) [w/o sync]}
         \label{fig:varying-buffer-zipf-1.0-lat}
     \end{subfigure}
     \hfill
  	\begin{subfigure}{0.24\linewidth}
  
         \centering
         \includegraphics[width=\linewidth]{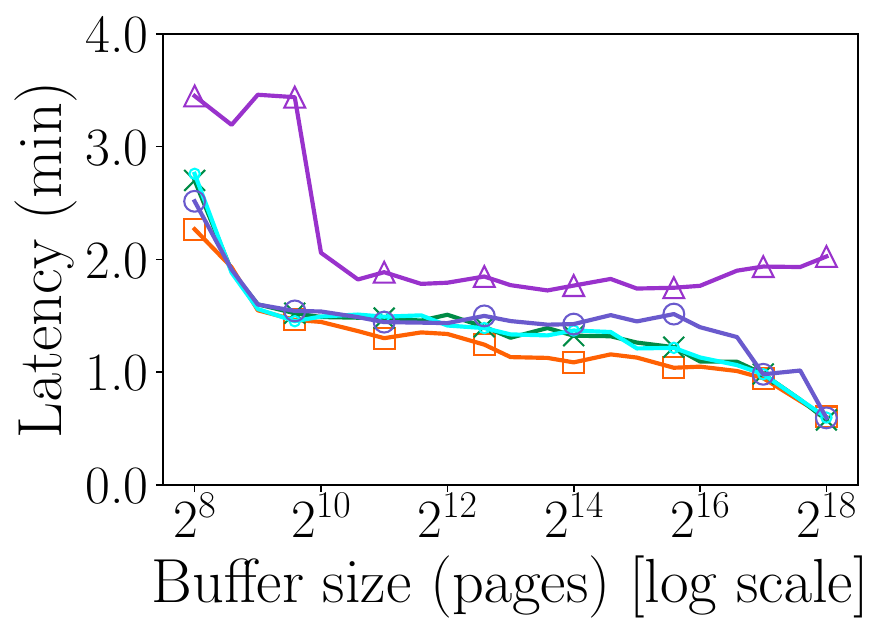}
         
         \caption{Zipf ($\alpha$=$0.7$) [w/o sync]}
         \label{fig:varying-buffer-zipf-0.7-lat}
     \end{subfigure}
     \hfill
     \begin{subfigure}{0.24\linewidth}
   
         \centering
         \includegraphics[width=\linewidth]{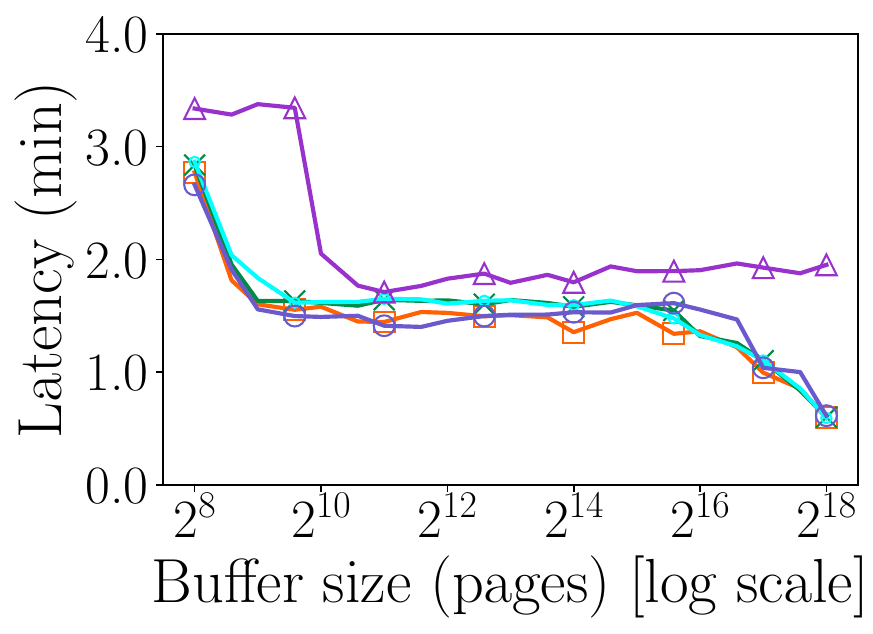}
         
         \caption{Uniform [w/o sync]}
         \label{fig:varying-buffer-uni-lat}
     \end{subfigure}
     \begin{subfigure}{0.24\linewidth}
         \centering
         \includegraphics[width=\linewidth]{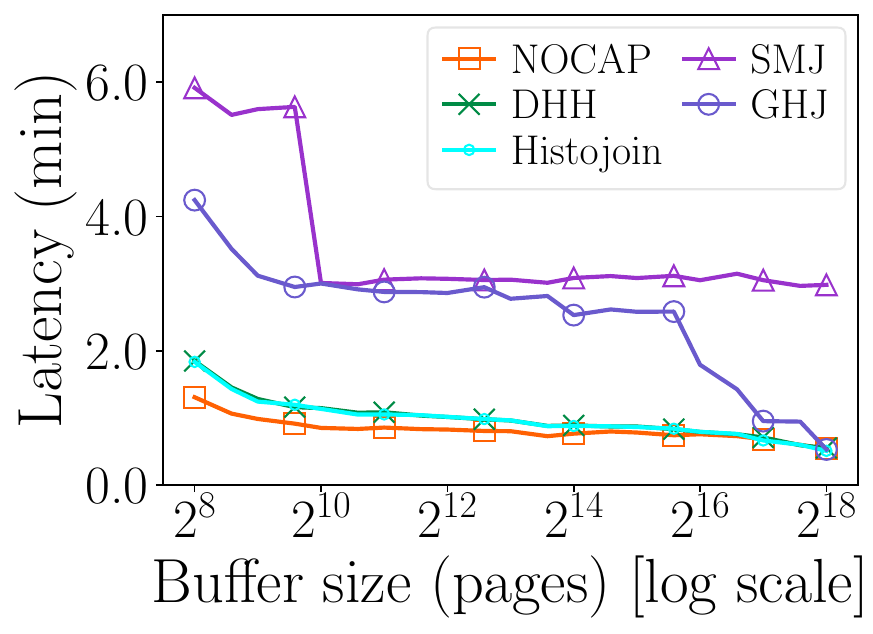}
        
         \caption{Zipf ($\alpha = 1.3$) [w/ sync]}
         \label{fig:varying-buffer-zipf-1.3-lat-sync}
          
     \end{subfigure}
     \hfill
      \begin{subfigure}{0.24\linewidth}
         \centering
         \includegraphics[width=\linewidth]{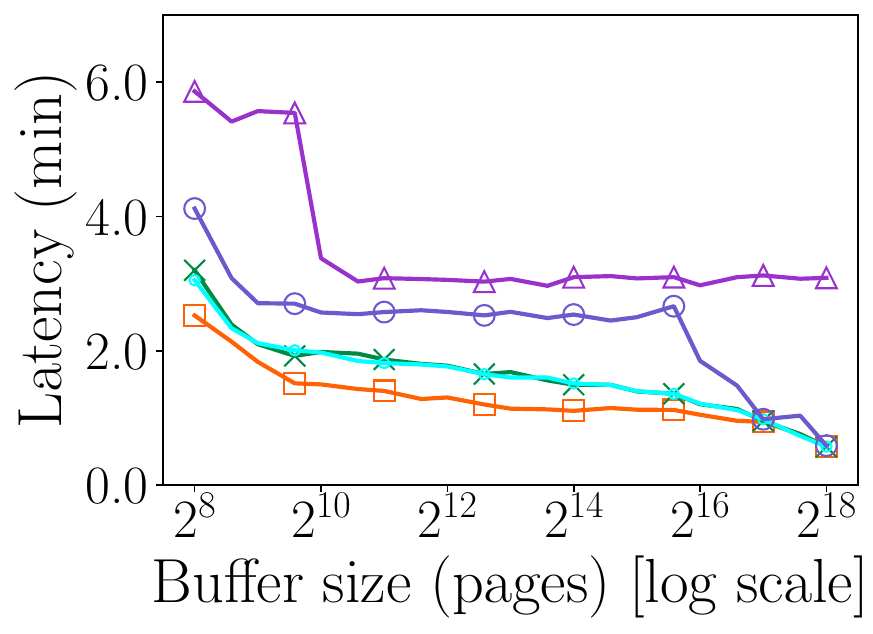}
        
         \caption{Zipf ($\alpha$=$1.0$) [w/ sync]}
         \label{fig:varying-buffer-zipf-1.0-lat-sync}
          
     \end{subfigure}
     \hfill
  	\begin{subfigure}{0.24\linewidth}
         \centering
         \includegraphics[width=\linewidth]{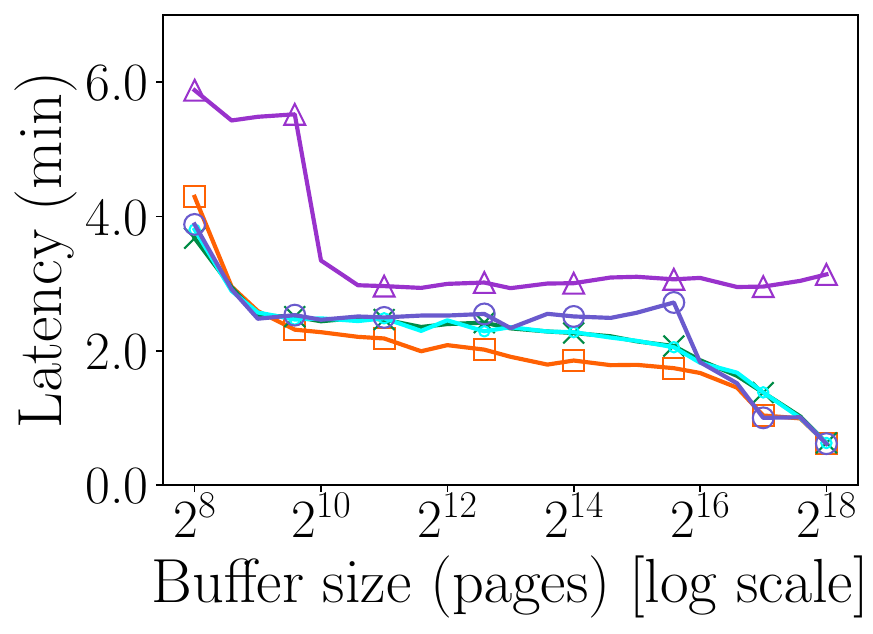}
         \caption{Zipf ($\alpha$=$0.7$) [w/ sync]}
         \label{fig:varying-buffer-zipf-0.7-lat-sync}
         
     \end{subfigure}
     \hfill
     \begin{subfigure}{0.24\linewidth}
         \centering
         \includegraphics[width=\linewidth]{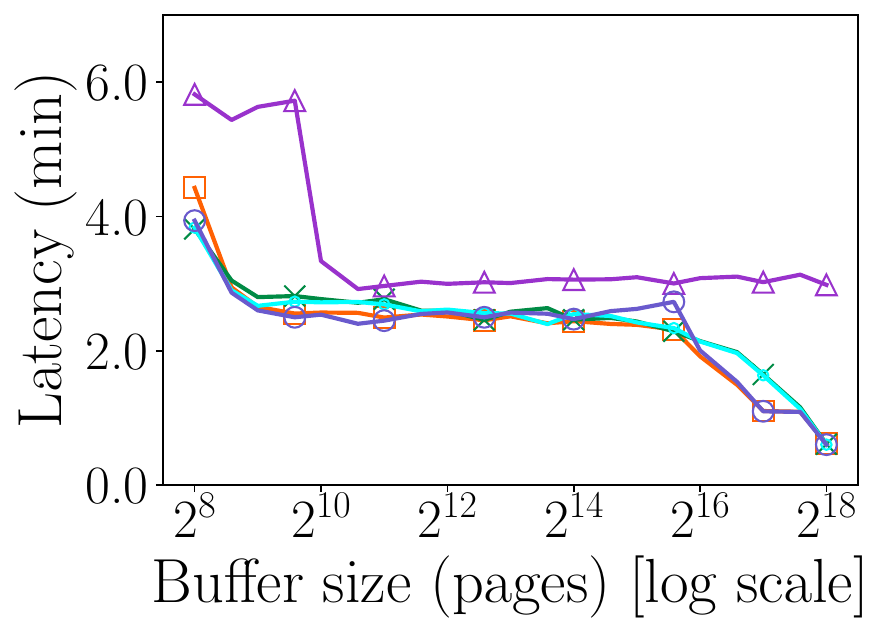}  
         \caption{Uniform [w/ sync]}
         \label{fig:varying-buffer-uni-lat-sync}
     \end{subfigure}
     \vspace{-0.1in}
     \caption{While state-of-the-art skew optimization in DHH helps reduce the I/O cost, \ApprAlgNameAbbr{} can better exploit the correlation skew to achieve even lower I/O cost and latency. The benefit of \ApprAlgNameAbbr{} is more pronounced with skew correlation.}
\label{fig:vary-skewness}
\end{figure*}

We now present experimental results for our practical method. 

\Paragraph{Approaches Compared.}
We compare \ApprAlgNameAbbr{} with Grace Hash Join (GHJ), Sort-Merge Join (SMJ), and Dynamic Hybrid Hash join (DHH).
For all the partitioning-based methods (GHJ, DHH, and ours), we apply a light optimizer that picks the most efficient algorithm according to Table~\ref{tab:cost_model} in the partition-wise join.
We also augment GHJ by allowing it to fall back to NBJ if the latter has a lower cost.
For DHH, we follow prior approaches that use fixed thresholds to trigger skew optimization (2\% of the available memory is used for an in-memory hash table for skew keys if the sum of their frequency is larger than 2\% of the relation size).
We also compare Histojoin (by setting the trigger frequency threshold as zero) as one of our baselines.
All the approaches are implemented in C++, compiled with gcc 10.1.0, in a CentOS Linux with kernel 4.18.0.
For \ApprAlgNameAbbr{}, DHH, and Histojoin, we assume top $k=50K$ frequently matching keys are tracked (i.e., 5\% of the keys when $n=1M$).
In PostgreSQL's implementation of DHH, the frequent keys are stored using a small amount of system cache, which does not consume the user-defined working memory budget unless they are inserted into the in-memory hash table.
Similarly, in our prototype, the top-$k$ keys are given.
Due to the pruning techniques (\S\ref{subsec:dynamic_programming}), computing the partitioning scheme (Alg.~\ref{alg:find_best_Mk}) with $k=50K$ takes less than one second, so we omit its discussion.

\subsection{Sensitivity Analysis}
\label{subsec:exp_synthetic_dataset}
We now experimental results with synthetic data produced by our workload generator that allows us to customize the distribution of matching keys between two input relation to test the robustness of different join methods.

\Paragraph{Experimental Setup.} We use our in-house server, which is equipped with 375GB memory in total and two Intel Xeon Gold 6230 2.1GHz processors, each having 20 cores with virtualization enabled. 
For our storage, we use a 350GB PCIe P4510 SSD with direct I/O enabled.
We can change the read/write asymmetry by turning the \verb|O_SYNC| flag on and off.
When \verb|O_SYNC| is on, every write is ensured to flush to disk before the write is completed, and thus has higher asymmetry ($\mu_{sync}=3.3$, $\tau_{sync}=3.2  $). 
When it is off, we have lower asymmetry $(\mu_{no\_sync}=1.28$, $\tau_{no\_sync}=1.2)$.
By default, sync I/O is off to accelerate joins.

\Paragraph{Experimental Methodology.}
We first experiment with a synthetic workload, which contains two tables $R$ and $S$ with $n_R=1M$ and $n_S=8M$.
The record size is $1$KB for both $R$ and $S$, and thus we have $\|R\|=250K$ pages and $\|S\|=2M$ pages (the page size is fixed as $4$KB in all the experiments).
In this experiment, we vary the buffer size from $0.5\cdot \sqrt{F\cdot \|R\|}\approx 256$ pages to $\|R\|=250K$ pages. 
We use \#I/Os (reads + writes) and latency as two metrics, and when comparing \#I/Os, we also run \ExactAlgName{} (\S\ref{subsec:hybrid-join}) to plot the optimal (lower bound) \#I/Os . 
Further, we run our experiments under uniform correlation and Zipfian correlation (short as Zipf) with $\alpha=0.7,1.0,1.3$.

\Paragraph{\ApprAlgNameAbbr{} Dominates for any Correlation Skew and any Memory Budget.}
For all the join methods, \#I/Os decreases as the buffer size increases, as shown Figures~\ref{fig:varying-buffer-zipf-1.3-io}-\ref{fig:varying-buffer-uni-io}.
We also observe that \ApprAlgNameAbbr{} achieves the lowest \#I/Os (near-optimal) among all the join algorithms for any correlation skew and any memory budget.
When we compare latency from Figures~\ref{fig:varying-buffer-zipf-1.3-lat}-\ref{fig:varying-buffer-uni-lat}, we also conclude that GHJ and SMJ have a clear gap while they have nearly the same \#I/Os regardless of the join correlation.
This is because random reads in SMJ are $1.2\times$ slower than sequential reads in GHJ.
While two traditional join methods have no optimization for skewed join correlations, DHH, Histojoin, and \ApprAlgNameAbbr{} can take advantage of correlation skew to alleviate \#I/Os.
However, DHH and Histojoin cannot fully exploit the data skew because it limits the space for high-frequency keys with a fixed threshold ($2\%$ of the total memory budget).
Compared to DHH and Histojoin, \ApprAlgNameAbbr{} normally has fewer \#I/Os because it freely decides the size of the in-memory hash table for frequent keys without any predefined thresholds.
Figures~\ref{fig:varying-buffer-zipf-1.0-lat}-\ref{fig:varying-buffer-uni-lat} 
show that the latency gap between \ApprAlgNameAbbr{} and DHH (Histojoin) is more pronounced when we have medium skew ($\alpha=1$ or $\alpha=0.7$), compared to the uniform correlation.
However, when the correlation is highly skewed ($\alpha=1.3$), the small hash table (2\% of the available memory) in DHH is enough to capture the most frequent keys and, hence, issue near-optimal \#I/Os (Figure~\ref{fig:varying-buffer-zipf-1.3-io}).
As a result, the benefit of \ApprAlgNameAbbr{} for $\alpha=1.3$ (Figure~\ref{fig:varying-buffer-zipf-1.3-lat}) is smaller than for $\alpha=1$ (Figure~\ref{fig:varying-buffer-zipf-1.0-lat}).
Nevertheless, \ApprAlgNameAbbr{} still issues 30\% fewer \#I/Os when the memory budget is small ($2^8$ pages).

\begin{figure*}[t]
\begin{minipage}{0.48\textwidth}
\begin{subfigure}{0.49\linewidth}
     \includegraphics[width=\linewidth]{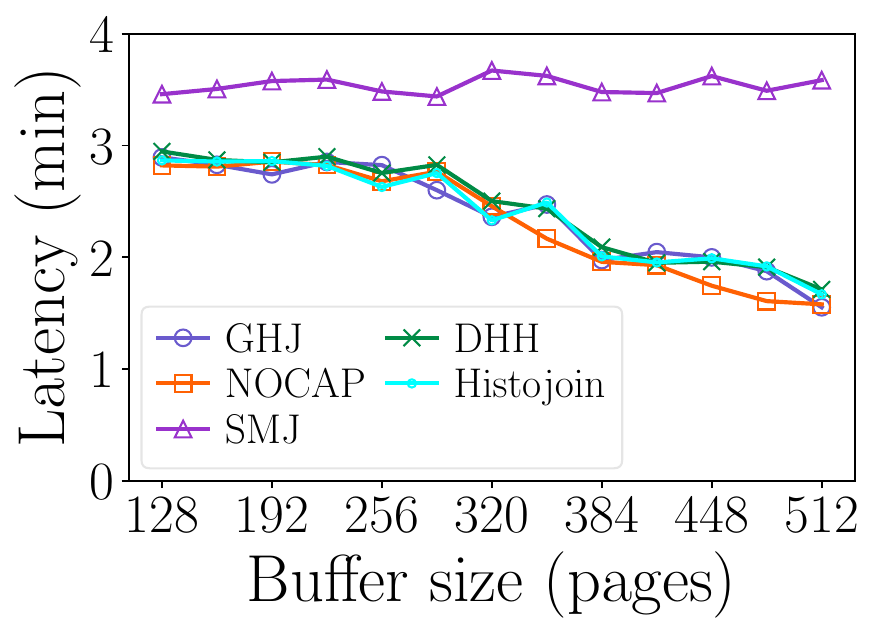}
     
         \caption{Uniform}
         \label{fig:tiny-mem-uni}
\end{subfigure}
\hfill
\begin{subfigure}{0.49\linewidth}
         \centering
         \includegraphics[width=\linewidth]{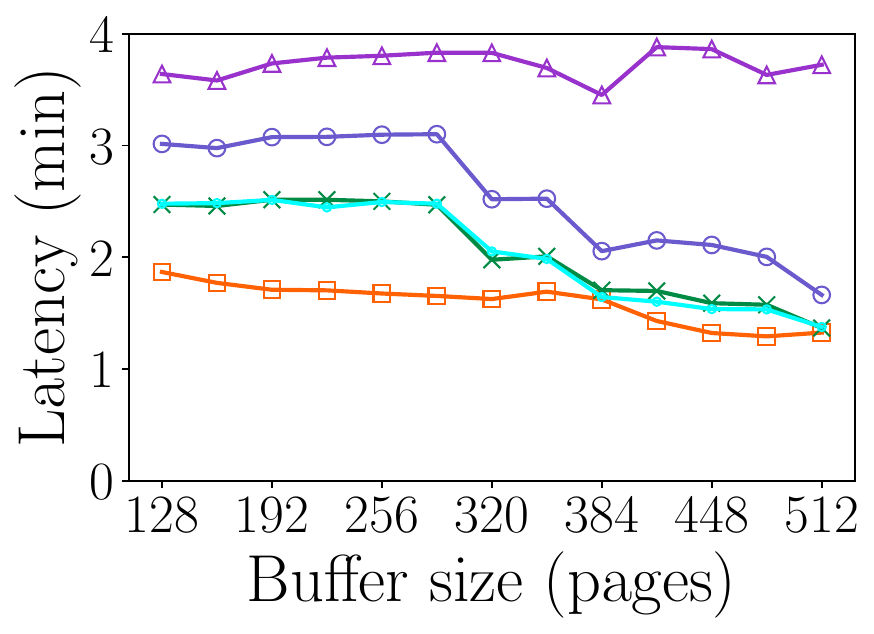}
         \caption{Zipf ($\alpha$=$1.0$)}
\label{fig:tiny-mem-zipf-alpha-1}
\end{subfigure}
\vspace{-0.15in}
\caption{When the memory is limited, \ApprAlgNameAbbr{} can even reach 10\% speedup with uniform workload.}
\label{fig:tiny-mem}
\end{minipage}
\hfill
\begin{minipage}{0.48\textwidth}
\centering
     \begin{subfigure}{0.49\linewidth}
\vspace{-1.16in}\includegraphics[width=\linewidth]{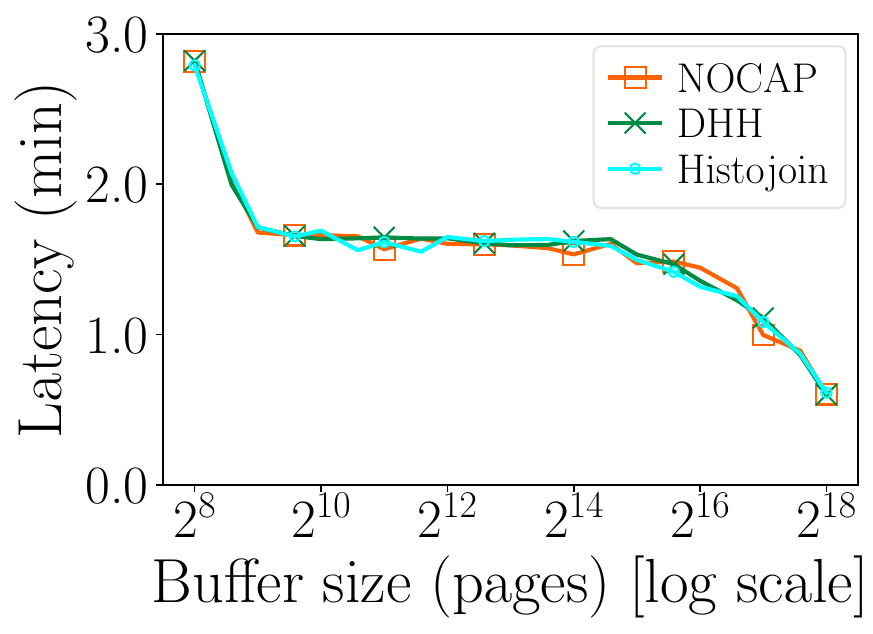}
         \caption{Uniform}
         \label{fig:noisy-uni}
     \end{subfigure}
     \hfill
  	\begin{subfigure}{0.49\linewidth}
         \centering
         \includegraphics[width=\linewidth]{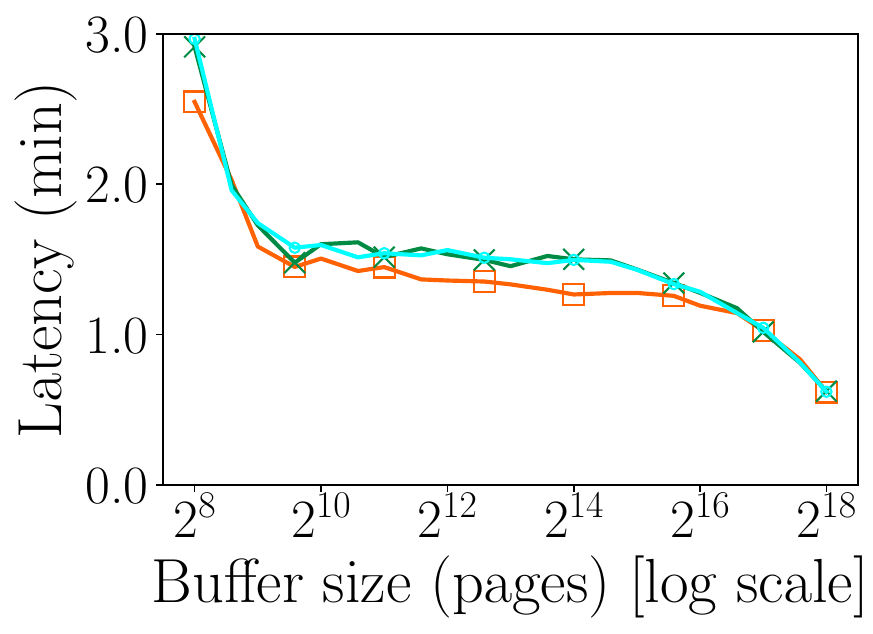}
         
         \caption{Zipf ($\alpha$=$0.7$)}
         \label{fig:noisy-zipf-alpha-0.7}
     \end{subfigure}
\vspace{-0.15in}
\caption{With noisy MCVs, both DHH and \ApprAlgNameAbbr{} have similar performance to Figures~\ref{fig:varying-buffer-uni-lat} and ~\ref{fig:varying-buffer-zipf-0.7-lat}.}
\label{fig:noisy-exp}
\end{minipage}
\end{figure*}

\begin{figure}[t]
    \centering
    \begin{subfigure}{0.3\linewidth}
         \centering
        \includegraphics[width=\linewidth]{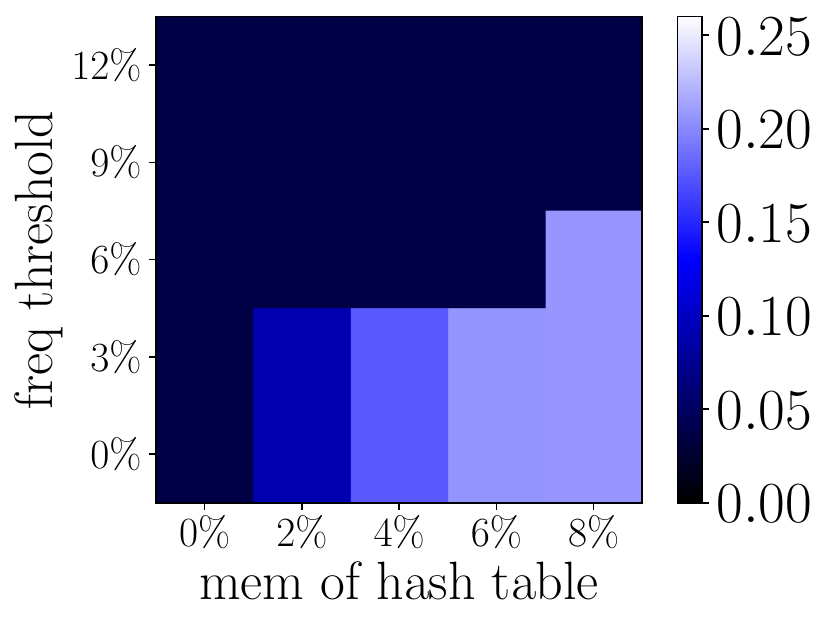}
         \caption{$B=2 MB$}
         \label{fig:DHH-vary-th-low-B-z0.7}
     \end{subfigure}
  	\begin{subfigure}{0.3\linewidth}
         \centering
        \includegraphics[width=\linewidth]{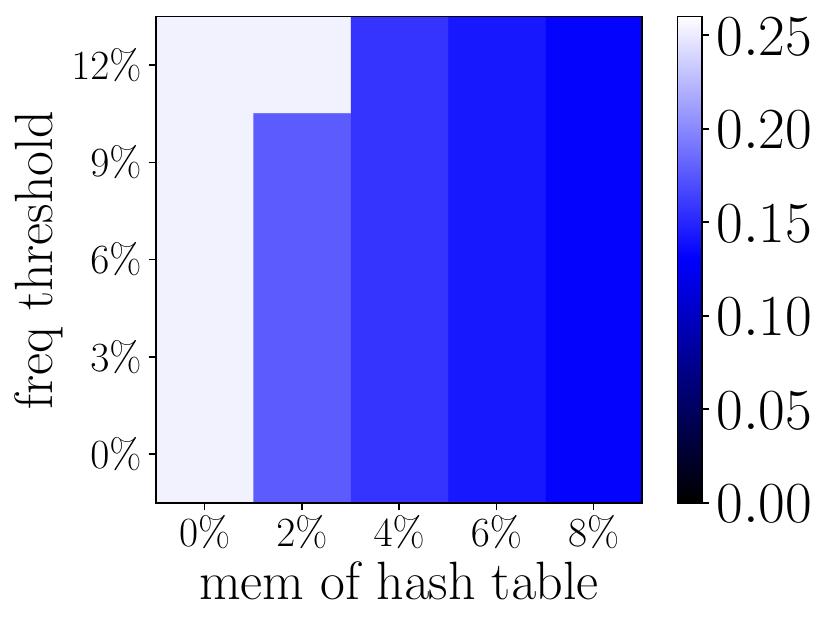}
        
         \caption{$B=32 MB$}
         \label{fig:DHH-vary-th-high-B-z0.7}
     \end{subfigure}
  \vspace{-0.12in}
     \caption{DHH requires careful tuning to achieve its best performance, which is still slower than \ApprAlgNameAbbr{}.}
\label{fig:vary-DHH-thresholds}
\end{figure}

\begin{figure*}[t]
\centering
     \begin{subfigure}{0.24\linewidth}
         \centering \includegraphics[width=\linewidth]{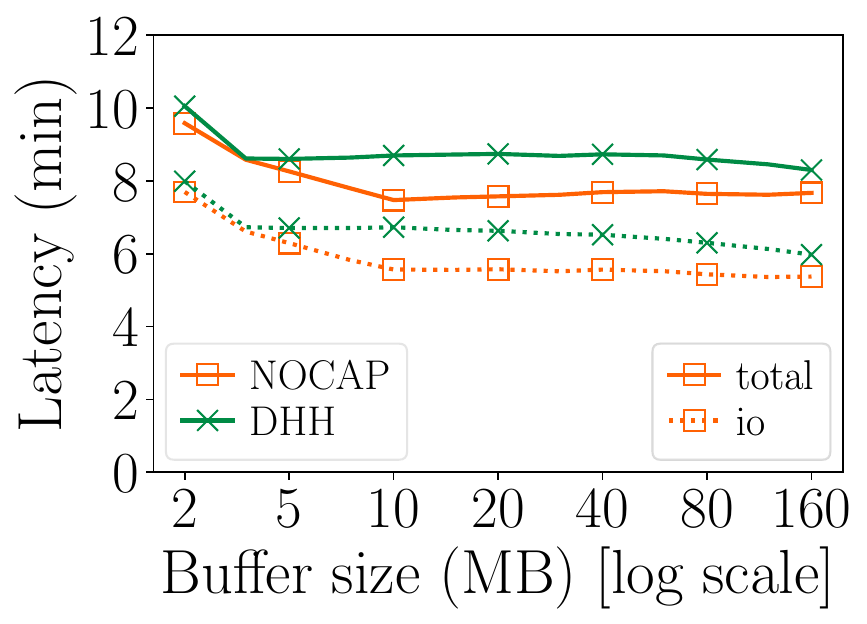}
          
         \caption{$\sigma_{S}$=$0.488, SF$=$10$}
         \label{fig:TPCH-uni}
     \end{subfigure}
     \hfill
  	\begin{subfigure}{0.24\linewidth}
         \centering
         \includegraphics[width=\linewidth]{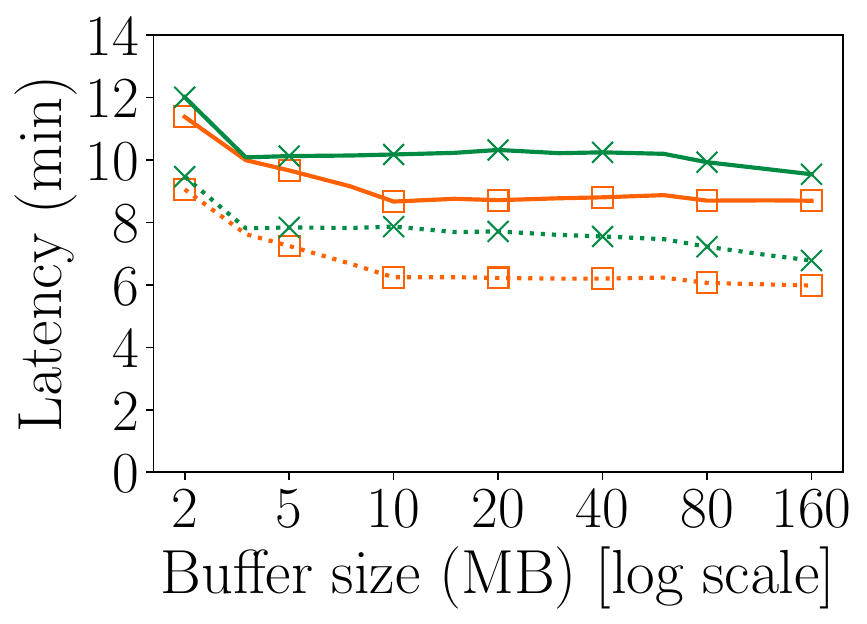}
         \caption{$\sigma_{S}$=$0.63, SF$=$10$}
         \label{fig:TPCH-skew}
     \end{subfigure}
    \hfill
     \begin{subfigure}{0.24\linewidth}
         \centering \includegraphics[width=\linewidth]{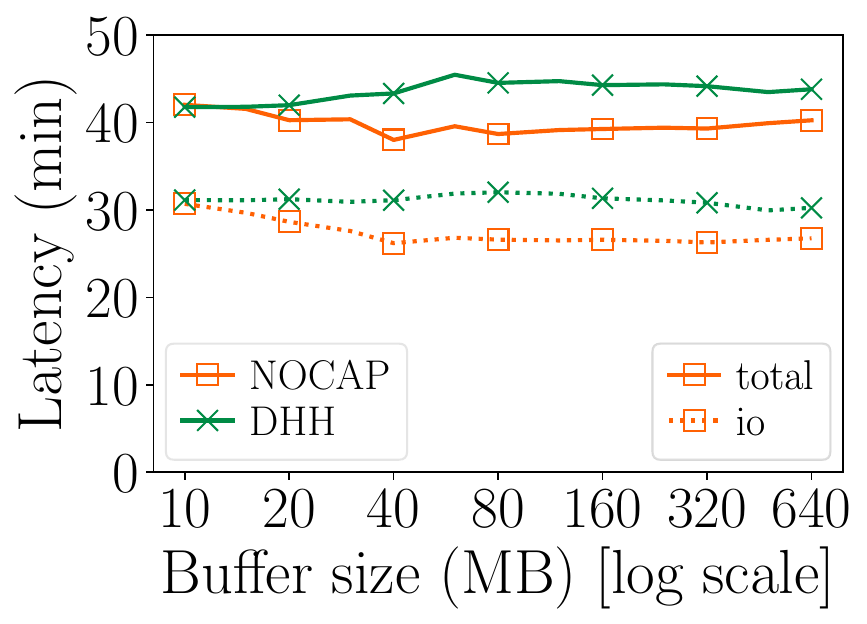}
        \caption{$\sigma_{S}$=$0.488, SF$=$50$}
         \label{fig:TPCH-uni-SF4}
     \end{subfigure}
     \hfill
  	\begin{subfigure}{0.24\linewidth}
         \centering
         \includegraphics[width=\linewidth]{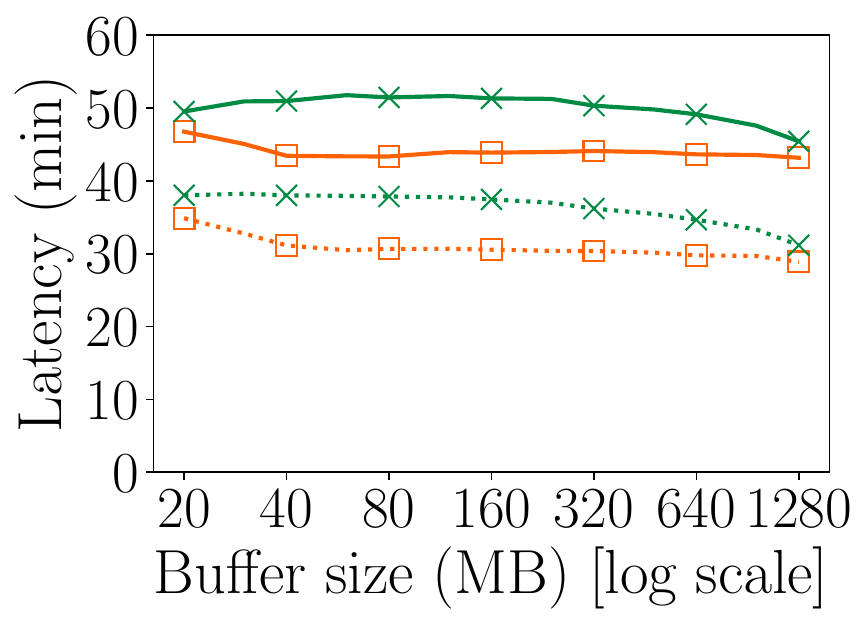}\caption{$\sigma_{S}$=$0.63, SF$=$50$}
         \label{fig:TPCH-skew-SF4}
     \end{subfigure}
\vspace{-0.1in}
\caption{TPC-H: \ApprAlgNameAbbr{} leads to higher speedup when more data from $S$ are joined in a query.}
\label{fig:TPCH}
\end{figure*}

\Paragraph{DHH Cannot Adapt to Different Memory Budget.} The skew optimization in DHH relies on two thresholds (memory budget and minimum frequency) and thus it cannot easily adapt to a different memory budget.
To verify this, we compare DHH with NOCAP under a given Zipfian correlation ($\alpha=0.7$) using as memory budgets $2$MB and $32$MB.
Figures~\ref{fig:DHH-vary-th-low-B-z0.7} and~\ref{fig:DHH-vary-th-high-B-z0.7} show the the percentage of reduced \#I/Os using NOCAP based on DHH when we vary the two thresholds used by DHH.
When $B=2$MB, the best DHH (the darkest cell) does not trigger skew optimization by setting a high frequency trigger threshold or the memory budget as zero.
In contrast, when $B=32$MB, the best DHH assigns 8\% memory to build the hash table for skewed keys.
Note that even though DHH can achieve close-to-\ApprAlgNameAbbr{} performance by varying these two thresholds, this relies on well-tuned parameters, which may not be feasible when the workload and the buffer size changes.

\Paragraph{\ApprAlgNameAbbr{} Outperforms DHH Even for Uniform Correlation with Small Memory Budget.}
We also examine the speedup when we have a lower memory budget.
Although Figures~\ref{fig:varying-buffer-uni-lat-sync} and~\ref{fig:varying-buffer-uni-lat} show that there \emph{seems} no difference between DHH and \ApprAlgNameAbbr{} when the workload is uniform, we actually capture a larger difference after we narrow down the buffer range (128$\sim$512 pages).
As shown in Figure~\ref{fig:tiny-mem-uni}, \ApprAlgNameAbbr{} can even achieve up to 15\% and 10\% speedup when there are 480 and 352 pages.
The step-wise pattern of DHH and GHJ comes from random (uniform) partitioning. 
When the memory is smaller than $\sqrt{\|R\|\cdot F}$, uniform partitioning easily makes each partition larger than the chunk size, as illustrated in Figure~\ref{fig:uniform-vs-non-uniform-partition}.
On the contrary, rounded hash allows some partitions to be larger so that other partitions have fewer chunks and thus require fewer passes to complete the partition-wise joins.
When the correlation becomes more skewed, \ApprAlgNameAbbr{} can be more than 30\% faster than DHH under limited memory, as shown in Figure~\ref{fig:tiny-mem-zipf-alpha-1}.

\Paragraph{DHH and \ApprAlgNameAbbr{} Both Exhibit Robustness for Noisy MCVs.}
To examine how noisy \CT{} can affect DHH and \ApprAlgNameAbbr{}, we add Gaussian noise to the \CT{} values, with the average noise as 0 (the value of the average noise has no impact over MCVs) and $\sigma=n_S/n_R$.
We reuse most of experimental settings so that $n_S/n_R=8$ and thus $\CT_{noise}[i] \in [\CT[i]-8, \CT[i]+8]$ with probability 68\%. 
We redo the experiments with uniform and Zipfian jon correlation ($\alpha=0.7$), and the results are shown in Figure~\ref{fig:noisy-exp}.
We do not observe a significant difference between the results with noisy $\CT$ values and the original results in Figures~\ref{fig:vary-skewness}.
The rational behind this is that keys with high frequency are still likely to be prioritized (cached) during hybrid partitioning even after we add the Gaussian noise, thus having very similar performance to the results without noise.

\subsection{Experiments with TPC-H, JCC-H, and JOB}
\label{subsec:exp_large_dataset}
We now experiment with TPC-H~\cite{TPCH}, JCC-H~\cite{Boncz2017}, and JOB~\cite{Leis2015}, using DHH as our main competitor.

\Paragraph{Experimental Setup.}
We experiment with a modified TPC-H Q12. Q12 selects data from table \verb|lineitem|, and joins them with \verb|orders| on \verb|l_orderskey=o_orderskey|, followed by  an aggregation.
As we have already seen, DHH has only a minor difference from \ApprAlgNameAbbr{} when the correlation is uniform. Here, we focus on a skew correlation in TPC-H data.
To achieve this, we modify the TPC-H dbgen codebase --
all the keys are classified into \emph{hot} or \emph{cold} keys, and the frequency of hot keys and cold keys, respectively, follows two different uniform distributions, which controls the overall skew and the average matching keys~\cite{Chatterjee2022,Bruck2006}.
In this experiment, we focus on a skew correlation where 0.5\% of the \verb|o_orderkey| keys match 500 \verb|lineitem| records on average and the rest keys only match 1.5 \verb|lineitem| records on average.

\begin{figure*}[t]
\centering
     \begin{subfigure}{0.24\linewidth}
         \centering \includegraphics[width=\linewidth]{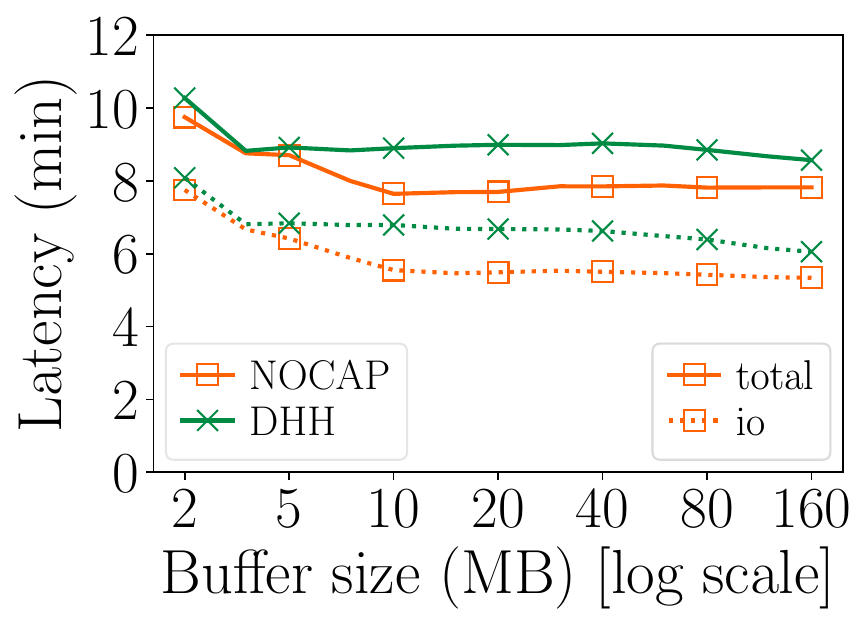}
         \caption{JCC-H (Tuned Skew)}
         \label{fig:JCC-H-tuned-skew}
     \end{subfigure}
     \hfill 
     \begin{subfigure}{0.24\linewidth}
         \centering \includegraphics[width=\linewidth]{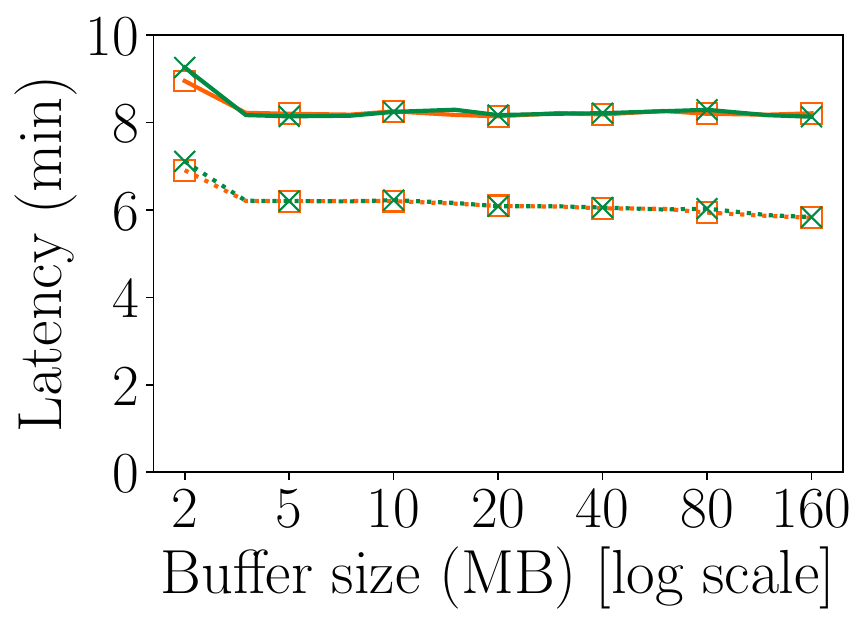}
         \caption{JCC-H (Original Skew)}
         \label{fig:JCC-H-origin-skew}
     \end{subfigure}
     \hfill
     \begin{subfigure}{0.24\linewidth}
         \centering
         \includegraphics[width=\linewidth]{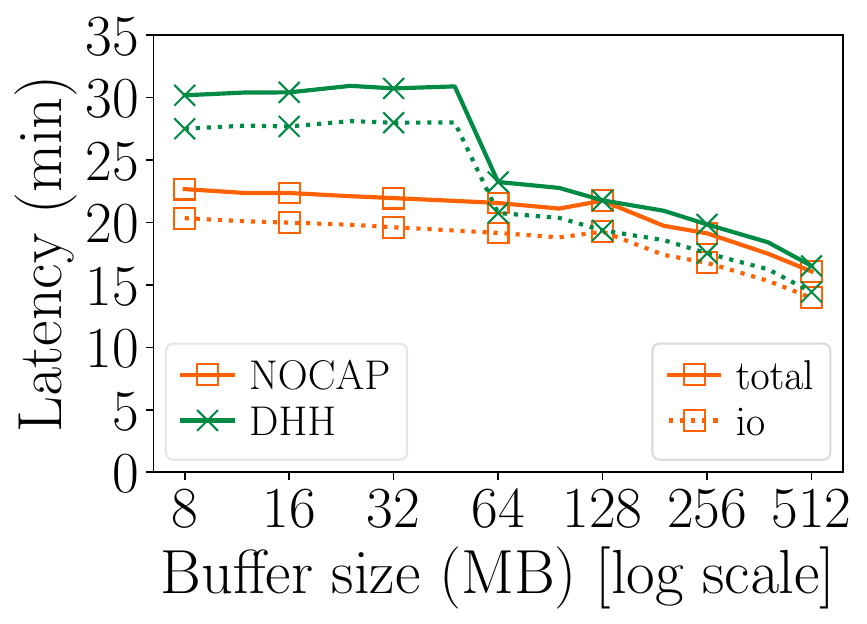}
         \caption{JOB (cast $\bowtie$ title)}
         \label{fig:JOB_cast_join_title}
     \end{subfigure}
     \hfill
  	\begin{subfigure}{0.24\linewidth}
         \centering
         \includegraphics[width=\linewidth]{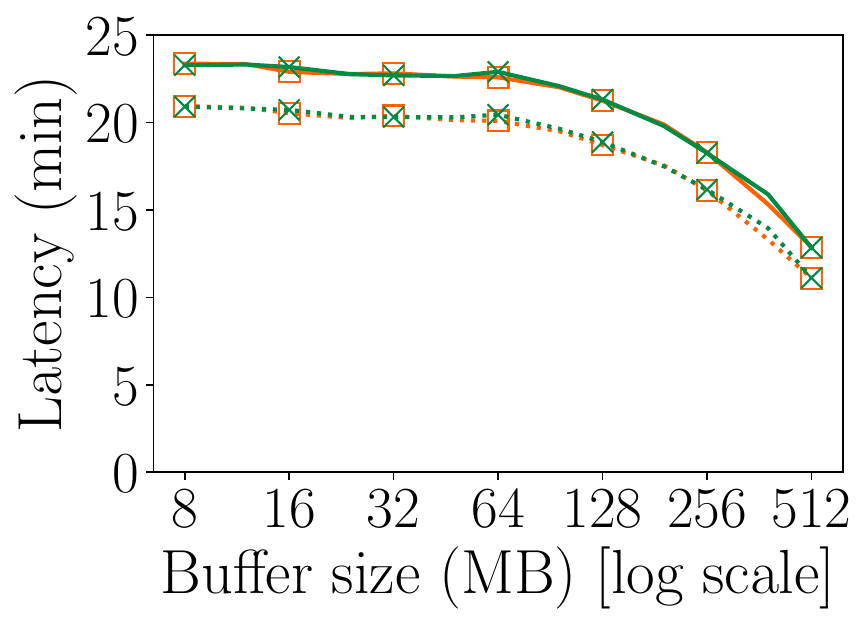}
         \caption{JOB (cast $\bowtie$ name)}
         \label{fig:JOB_cast_join_name}
     \end{subfigure}
    \vspace{-0.12in}
    \caption{While DHH can perform as close as \ApprAlgNameAbbr{}, \ApprAlgNameAbbr{} is more adaptive when the workload varies.}
\label{fig:other-dataset}
\end{figure*}

We further remove the filtering condition on \verb|l_shipmode| and \verb|l_receiptdate| so that the size of the filtered \verb|lineitem| is larger than \verb|orders| (since we focus on the case when $\|R\| < \|S\|$).
We then have two conditions left, \verb|l_receiptdate|<\verb|l_shipdate|
 and \verb|l_shipdate|<\verb|l_commitdate|, which has separately selectivity $0.488$ and $0.63$.
 We choose one of filter conditions to vary the selectivity.
We run our experiments with these two selectivity and two scale factors ($SF$=$10$, $SF$=$50$), as shown in Figure~\ref{fig:TPCH}.
In these experiments, we employ storage-optimized AWS instances (\textit{i3.4xlarge}) and we turn off \verb|O_SYNC|. The device asymmetry is $\mu=1.2$ and $\tau=1.14$.

\Paragraph{\ApprAlgNameAbbr{} Accelerates TPC-H Joins.}	We observe that the difference in the total latency between \ApprAlgNameAbbr{} and DHH is not so large as we see in previous experiments on skew correlation in Figure~\ref{fig:TPCH-uni}.
In fact, we find that the proportion of time spent on I/Os is lower in the TPC-H experiment due to extra aggregations in Q12.
In earlier experiments (e.g., Figures~\ref{fig:varying-buffer-zipf-1.0-lat} and ~\ref{fig:varying-buffer-zipf-0.7-lat}), the time spent on CPU is 10$\sim$20 seconds out of a few minutes (the total latency). 
In contrast, we observe that in both Figures~\ref{fig:TPCH-uni} and ~\ref{fig:TPCH-skew}, the proportion of time spent on I/Os is less than 85\%.
This explains why the speedup of \ApprAlgNameAbbr{} diminishes when we have more CPU-intensive operations (e.g., aggregation) in a query.
In addition, when we adjust the selectivity so that more data from \verb|lineitem| are involved in the join, we observe a higher benefit in Figure ~\ref{fig:TPCH-skew}.
We have similar observation when we move to a larger data set (SF=50), as shown in Figures~\ref{fig:TPCH-uni-SF4} and~\ref{fig:TPCH-skew-SF4}.
In fact, in both cases, when we have larger \verb|lineitem| data involved in the join, \ApprAlgNameAbbr{} usually leads to higher speedup.

\Paragraph{\ApprAlgNameAbbr{} Outperforms DHH for JCC-H and JOB.} To further validate the effectiveness of \ApprAlgNameAbbr{}, we experiment with JCC-H (Join Cross Correlation) and JOB (Join Order Benchmark).
JCC-H is built on top of TPC-H with adding join skew, so we reuse our TPC-H experimental setup (running the modified Q12 with $SF=10$ and $\sigma_S=0.488$).
In the original JCC-H dataset, the majority (i.e.,~$300K\cdot SF$) of \verb|lineitem| records only join with 5 \verb|orders| records.
To vary the skew, we tune the code by allowing around $5100\cdot SF$ \verb|orders| records to match 600 \verb|lineitem| records on average.
In addition to JCC-H, we also compare DHH and \ApprAlgNameAbbr{} by executing a PK-FK join between two tables from the JOB dataset.
Specifically, \verb|cast_info| stores the cast info between movies and actors, which are respectively stored in tables \verb|title| and \verb|name|.
The number of movies per actor is highly skewed where the top 50 actors match 0.6\% records in \verb|cast_info|, while the number of actors per movie is less skewed where the top 50 movies match less than 0.1\% records in \verb|cast_info| (``top'' in terms of the number of occurrences in \verb|cast_info|).
As we can observe in Figures~\ref{fig:other-dataset}a-d, when the correlation is extremely skewed (i.e., the original JCC-H and \verb|cast_info| $\bowtie$ \verb|name|), DHH with fixed thresholds can identify the skew and exploit it to achieve close-to-\ApprAlgNameAbbr{} performance. However, when it comes to medium skewed (i.e., tuned JCC-H and \verb|cast_info| $\bowtie$ \verb|title|), DHH with fixed thresholds results in suboptimal performance compared to \ApprAlgNameAbbr{}.

%% file: 6-discussion.tex
\section{Discussion}
\label{sec:discussion}
\Paragraph{General (Many-to-Many) Joins.} 
In many-to-many joins, we have two $\CT$ lists that store the frequency per key in each relation, and thus we need to create a new cost function (with two $\CT$ lists) to use as the main objective function.
Despite that, we can still reuse our dynamic programming algorithm~\ref{alg:matrix_dp} and replace $\textsf{CalCost}$ by considering extra $\CT$ values, however, the error bound is no longer guaranteed since our main theorem does not hold.
In practice, even this crude approximation of the optimal partitioning for general joins is worth exploring since the CPU cost of enumerating partitioning schemes is very low, and it covers the design space of DHH with variable values for its two thresholds. As a result, \ApprAlgNameAbbr{} is still expected to outperform DHH for general joins. We leave the complete formulation for many-to-many joins as future work.

\Paragraph{On-the-fly Sampling.} On-the-fly sampling can give us more accurate \CT~ values because it takes into account the applied predicate while MCVs rely only on the estimated selectivity.
However, relying on sampling without knowledge of MCVs makes the partitioning operation not pipelineable with the scan/select operator.
This is because partitioning relies on the exact join correlation, but in case of sampling we need one complete pass first to acquire this information.
In this case, we may downgrade \ApprAlgNameAbbr{} into DHH to avoid the additional pass.
Another approach proposed by Flow-join~\cite{Rodiger2016} is to employ a streaming sketch to obtain an approximate distribution during partitioning.
However, this requires scanning $S$ first to build the histogram, which conflicts with partitioning $R$ first in the hybrid join.
On the other hand, if we already have MCVs (note that it is acceptable to be noisy), sampling on-the-fly to enable sideways information passing (SIP)~\cite{Orr2019} can be helpful.
Specifically, we can sample $R$ on-the-fly during partitioning and build a Bloom Filter, which is later used when partitioning/scanning $S$. The additional Bloom Filter makes the estimated \CT~ more accurate, and thus the partitioning generated by \ApprAlgNameAbbr{} can get closer to expected \#I/O cost.
However, adding a BF in the memory also introduces new trade-offs between performance and memory consumption. 

%% file: 7-related-work.tex
\section{Related Work}
\label{sec:related_work}

\Paragraph{In-Memory Joins.} When the memory is enough to store the entire relations, in-memory join algorithms can be applied~\cite{Shatdal1994,Liu2015,Schuh2016,Pohl2020,Bandle2021,Barber2014}.
Similar to storage-based joins, in-memory join algorithms can be classified as hash-based and sort-based (e.g., MPSM~\cite{Albutiu2012} and MWAY~\cite{Balkesen2013a}).
Hash joins can be further classified into partitioned joins (e.g., parallel radix join~\cite{Balkesen2013a}) and non-partitioned joins (e.g., CHT~\cite{Balkesen2013}).
GPUs can also be applied to accelerate in-memory joins~\cite{Rui2015,Guo2019,Lutz2020}.
Although storage-based joins focus on \#I/Os, when two partitions both fit in memory after the partitioning phase, NOCAP can still benefit from optimized in-memory joins.

\Paragraph{Distributed Equi-Joins.} In a distributed environment, the efficiency of equi-joins also relies on load balance and communication cost.
When the join correlation distribution is skewed, uniform partitioning may overwhelm some workers by assigning excessive work to do~\cite{Walton1991}.
More specifically, when systems like Spark~\cite{Zaharia2012} rely on in-memory computations, the overall performance drops when the transmitted data for a machine do not fit in  memory.
Several approaches~\cite{Xu2008,Li2018,Beame2014,Rodiger2016} are proposed to tackle the data (and thus, workload) skew for distributed equi-joins.
 

%% file: 8-conclusion.tex
\section{Conclusion}
In this paper, we propose a new cost model for partitioning-based PK-FK joins that allows us to find the optimal partitioning strategy assuming accurate knowledge of the join correlation. 
Using this optimal partitioning strategy, we show that the state-of-the-art Dynamic Hybrid Hash Join (DHH) yields sub-optimal performance since it does not fully exploit the available memory and the correlation skew. 
To address DHH's limitations, we develop a practical near-optimal partitioning-based join with variable memory requirements (\ApprAlgNameAbbr{}) which is based on  \ExactAlgName{}. We show that \ApprAlgNameAbbr{} dominates the state-of-the-art for any available memory budget and for any join correlation, leading to up to 30\% improvement.

%% file: 8-appendix.tex
\section{Complexity Analysis for Monotonicity-based Pruining}
Based on two pruning inequalities (i.e., \textbf{Pruning 1} and \textbf{Pruning 2}) from weakly-ordered property, the range for $k$ is at most $\frac{i}{j}-\left\lfloor \frac{n-i-1}{m-j}\right\rfloor - 1 + (2 - \frac{1}{j})\cdot c_R$. 
If $\frac{i}{j}-\left\lfloor \frac{n-i-1}{m-j}\right\rfloor - 1 + (2-\frac{1}{j})\cdot c_R> \frac{i}{j}$, we can even skip the specific $j$ (line 7 in Algorithm~\ref{alg:matrix_dp}) because we do not have to iterate any $k$ in this case.
Therefore, the complexity for the third loop (line $7\sim 11$) is 
\begin{equation}
	\begin{split}
		& \max\left\{\frac{i}{j}-\left\lfloor \frac{n-i-1}{m-j}\right\rfloor - 1 + \left(2 - \frac{1}{j}\right)\cdot c_R, 0\right\}
		\leq \max\left\{\frac{i}{j}-\left\lfloor \frac{n-i-1}{m-j}\right\rfloor - 1, 0\right\} + \left(2 - \frac{1}{j}\right)\cdot c_R
	\end{split}
\end{equation}

Summing up the complexity for all possible $i$ and $j$, we have:

\begin{equation}
\begin{split}
\sum\limits_{i=1}^{n} \sum\limits_{j=2}^{m-1}\max\left\{\frac{i}{j}-\left\lfloor \frac{n-i-1}{m-j}\right\rfloor - 1, 0\right\} 
\leq & \sum\limits_{i=1}^{n}\sum\limits_{j=2}^{m-1}\max\left\{\frac{i}{j}- \frac{n-i}{m-j}, 0\right\} = \sum\limits_{j=2}^{m-1}\sum\limits_{i=\frac{j\cdot n}{m}}^{n} \frac{i\cdot m - n\cdot j}{j\cdot (m-j)} \\
= & \sum\limits_{j=2}^{m-1}\frac{\frac{n\cdot j+n\cdot m}{2}\cdot\left(n - \frac{n \cdot j}{m}+1\right)- n \cdot j\left(n - \frac{n\cdot j}{m}+1\right)}{j\cdot (m-j)} \\
= & \sum\limits_{j=2}^{m-1}\frac{\frac{n\cdot (m-j)}{2}\cdot \left(n - \frac{n \cdot j}{m}+1 \right) }{j\cdot (m-j)} = \sum\limits_{j=2}^{m-1}\frac{n\cdot \left(n - \frac{n \cdot j}{m}+1\right)}{2j} \\
= & \frac{1}{2}\left(\sum\limits_{j=2}^{m-1}n^2\cdot \left(\frac{1}{j} - \frac{1}{m}\right)+\sum\limits_{j=2}^{m-1}\frac{n}{j}\right) = O(n^2 \log m) \\
\sum\limits_{i=1}^{n}  \sum\limits_{j=2}^{m-1} \left(2 - \frac{1}{j}\right)\cdot c_R \leq \sum\limits_{i=1}^{n} 2 \cdot (m-2) \cdot c_R & = O(n m^2)
\end{split}
\label{eq:monotonicity-complexity-analysis}
\end{equation}
We typically consider $c_R \propto m (m << n)$ in our paper (the number of buffer pages is much smaller than the number of input records to be partitioned, otherwise we can simply put all the records from $R$ into buffer to execute \textsf{NBJ} without partitioning), and thus the computation complexity is dominated by $O(n^2 \log m)$.

\section{Proof of Weakly-Ordered and Divisible Property}
\Paragraph{Weakly-Ordered Property.} The weakly-ordered property can be formalized using Corollary~\ref{ceiling_size_decreasing_corollary}.
\begin{corollary}
	There exists an optimal partitioning scheme $P_{opt}$ where all the partitions contain consecutive elements and the sequence $\left\lceil \frac{\|\mathcal{N}(i)\|}{c_R}\right\rceil$ is in decreasing order.
	\label{ceiling_size_decreasing_corollary}
\end{corollary}

In fact, we can have a stronger Corollary ~\ref{size_decreasing_corollary} as follows.
\begin{corollary}
	\label{size_decreasing_corollary}
	There exists an optimal partitioning scheme $P$ where all the partitions contain consecutive elements w.r.t. their $CT$ value, and the partition size sequence $\{\|\mathcal{N}(i)\|\}$ is in decreasing (\textit{monotonic non-increasing}) order.
\end{corollary}

\begin{proof}
	We use the optimal partitioning scheme derived from the consecutive property. 
	Suppose that we find two adjacent partitions that are not in the non-decreasing order, as shown in Figure~\ref{fig:size_decreasing_partitions_proof}. 
	We can apply \textbf{Lemma~\ref{lem:swap}} multiple times until we have the larger partition is completely placed before the smaller one. 
	We can repeat this process for every two adjacent partitions until all the partitions are decreasingly sorted by their size. 
	Note that, we always ensure that the join cost does not increase by Lemma~\ref{lem:swap} and thus the final partitioning scheme is still optimal. The proof of Corollary~\ref{ceiling_size_decreasing_corollary} follows a similar procedure as above. Proof completes.
	\begin{figure}[t]
		\includegraphics[width=0.48\textwidth]{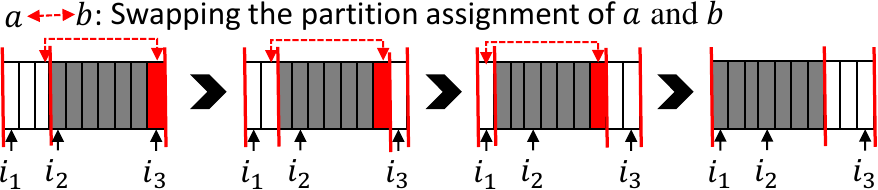}
		\vspace{-0.15in}
		\caption{When there are two adjacent partitions which do not follow the decreasing order, we can always apply \textbf{Lemma~\ref{lem:swap}} to swap elements. 
			For example, in the first step, we have $i_3 > i_2 - 1$ and $\|\mathcal{N}(i_3)\| > \|\mathcal{N}(i_2 - 1)\|$, and thus 
			$\left\lceil \frac{\|\mathcal{N}(i_3)\| }{c_R} \right\rceil \geq \left\lceil \frac{\|\mathcal{N}(i_2-1)\|}{c_R} \right\rceil$, we can swap $i_3$ and $i_2-1$. 
			We continue swapping elements until the indexes of all the elements in the larger partition are smaller than the other partition.}
		\label{fig:size_decreasing_partitions_proof}
		\vspace{-0.1in}
	\end{figure}
\end{proof}

\Paragraph{Divisible Property.} We prove the divisible property using Corollary~\ref{divisible_corollary}.
\begin{corollary}
	\label{divisible_corollary}
	There exists an optimal partitioning $P_{opt}$ where the size of all partitions except the largest partition, is divisible by $c_R$. 
\end{corollary}
Note that \textbf{Corollaries ~\ref{divisible_corollary}} and \textbf{\ref{size_decreasing_corollary}} may give us two different partitioning strategies that have the same -- optimal -- cost.
In other words, an optimal solution might satisfy either one of the two or both. 
For example, consider $n = c_R \cdot (B-2) + r \; (0 < r < c_R )$. According to Corollary ~\ref{size_decreasing_corollary}, we will group the last $r$ records to form the smallest partition and the remaining records form $B-2$ partitions. However, The consecutive property shows that we need to group the first $r$ records together. 
Even though we may end up having a different partitioning scheme, the optimal cost remains the same, due to the consecutive property, Corollary ~\ref{divisible_corollary}, and ~\ref{size_decreasing_corollary}.
This explains why we use the weakly-ordered property (Corollary~\ref{ceiling_size_decreasing_corollary}) instead of the ordered one (Corollary~\ref{size_decreasing_corollary}) in our main paper.

\begin{figure}[ht]
	\includegraphics[width=0.48\textwidth]{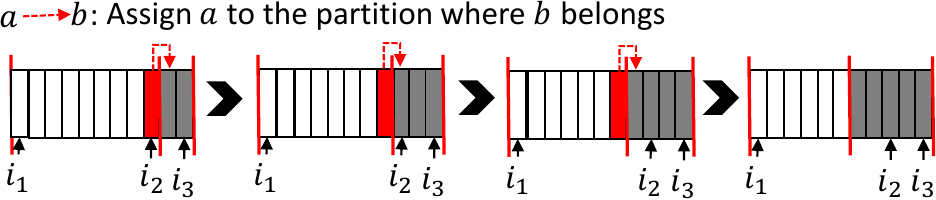}
	\vspace{-0.15in}
	\caption{When $\|\mathcal{P}_j\|$ is not divisible by $c_R$, we can move elements from $\mathcal{P}_{j-1}$ to $\mathcal{P}_j$ until $\left\lceil \frac{\|\mathcal{P}_{j-1}\|}{c_R}\right\rceil < \left\lceil \frac{\|\mathcal{P}_j\|}{c_R}\right\rceil$ or $\|\mathcal{P}_j\|$ is divisible by $c_R$. 
		In the above example, we note the partition containing $[i_2+1,i_3]$ as $\mathcal{P}_j$ and the previous partition containing $[i_1,i_2]$ as $\mathcal{P}_{j-1}$. Assuming $c_R = 5$, the above process stops when $\|\mathcal{P}_j\|=k\cdot 5)$ where $k \in \mathbf{N}$.}
	\label{fig:divisible_thereom_proof}
\end{figure}
\begin{proof}
	We aim to find a process that can obtain another optimal solution which satisfies both Corollaries \ref{divisible_corollary} and \ref{ceiling_size_decreasing_corollary}. 
	We note $\mathcal{P}_1,\mathcal{P}_2,...,\mathcal{P}_m$ as the partition sequence obtained from Corollary \ref{ceiling_size_decreasing_corollary}. 
	As such, these partitions are sorted by $\left\lceil \frac{\|\mathcal{P}_j\|}{c_R}\right\rceil$ in a descending order. 
	Now we select the largest $j (j \in [m])$ where $\|\mathcal{P}_j\|$ is not divisible by $c_R$. 
	If $j = 1$, we do nothing as the partitioning scheme already satisfies Corollary~\ref{divisible_corollary}. 
	Otherwise, we move adjacent elements from $\mathcal{P}_{j-1}$ to $\mathcal{P}_j$ until $\left\lceil \frac{\|\mathcal{P}_{j-1}\|}{c_R}\right\rceil < \left\lceil \frac{\|\mathcal{P}_j\|}{c_R}\right\rceil$  or $\|\mathcal{P}_j\|$ is a multiple of $c_R$. 
	When $\left\lceil \frac{\|\mathcal{P}_{j-1}\|}{c_R}\right\rceil < \left\lceil \frac{\|\mathcal{P}_j\|}{c_R}\right\rceil$, we swap elements according to the process mentioned in the proof of Corollary~\ref{size_decreasing_corollary} so that the decreasing order of $\left\lceil \frac{\|\mathcal{P}_j\|}{c_R}\right\rceil$ is maintained. 
	In other words, we obtain two consecutive partitions $\mathcal{P}^{new}_{j-1}, \mathcal{P}^{new}_{j}$  and we have $\|\mathcal{P}^{new}_{j-1}\| = \|\mathcal{P}^{old}_{j}\|, \|\mathcal{P}^{new}_{j}\|=\|\mathcal{P}^{old}_{j-1}\|$.
	We can repeat this process until the largest $j$ of $\|\mathcal{P}_j\|$ can be divisible by $c_R$. 
	Note that, when $\left\lceil \frac{\|\mathcal{P}_{j-1}\|}{c_R}\right\rceil < \left\lceil \frac{\|\mathcal{P}_j\|}{c_R}\right\rceil$, after we swap elements between $\mathcal{P}_{j-1}$ and $\mathcal{P}_j$, $\mathcal{P}_{j-1}$ will be selected to add elements from $\mathcal{P}_{j-2}$. 
	The above process already ensures that the final partition contains consecutive elements and the descending order elements based on their CT value. 
	In this process, when we move elements to $\mathcal{P}_j$, $\left\lceil \frac{\|\mathcal{P}_j\|}{c_R}\right\rceil$ does not change since $\|\mathcal{P}_j\|$ is not divisible by $c_R$. 
	Therefore, after we move one element $i$ to $\mathcal{P}_j$, all the elements from $\mathcal{P}_j$ contribute the same cost and all the elements (except $i$) from $\mathcal{P}_{j-1}$ have equivalent or even lower cost since $\|\mathcal{P}_{j-1}\|$ decreases. 
	For element $i$, the cost cannot increase since we maintain the descending order of $\left\lceil \frac{\|\mathcal{P}_j\|}{c_R}\right\rceil$. Specifically, when we move element $i$ from $\mathcal{P}_{j-1}$ to $\mathcal{P}_j$, $\left\lceil \frac{\|\mathcal{N}(i)\|}{c_R}\right\rceil$ must decrease or remain the same. Proof completes. 
\end{proof}